\documentclass[journal,twocolumn]{IEEEtran}
\usepackage{graphicx}
\usepackage{cite} 
\usepackage{subfigure}
\usepackage{float}
\usepackage{amsmath}
\usepackage{indentfirst}
\usepackage{algorithm}
\usepackage{algorithmic}
\usepackage{multirow}
\usepackage{wasysym}
\usepackage{amssymb}
\usepackage{url}
\usepackage{enumitem}
\usepackage{caption}
\usepackage{bm}
\usepackage{amsthm}
\usepackage{todonotes}
\usepackage{multicol}
\usepackage{multirow}
\usepackage{hyperref}

\newtheorem{lem}{Lemma}
\newtheorem{prop}{Proposition}
\newtheorem{cor}{Corollary}

\theoremstyle{definition}	
\newtheorem{defn}{Definition}

\theoremstyle{remark}
\newtheorem{rem}{Remark}

\newcommand{\tabincell}[2]{\begin{tabular}{@{}#1@{}}#2\end{tabular}}
\begin{document}
	\author{
		\IEEEauthorblockN{Qi Yu\IEEEauthorrefmark{1}\IEEEauthorrefmark{2}, Wei Dai\IEEEauthorrefmark{2}, Zoran Cvetkovi\'c\IEEEauthorrefmark{3}, Jubo Zhu\IEEEauthorrefmark{1}}\\
		\IEEEauthorblockA{\IEEEauthorrefmark{1}College of Liberal Arts and Sciences, National University of Defense Technology, Changsha, China}\\
		\IEEEauthorblockA{\IEEEauthorrefmark{2}Department of Electrical and Electronic Engineering, Imperial College London, UK}\\
		\IEEEauthorblockA{\IEEEauthorrefmark{3}Department of Informatics, King’s College London, UK}
		\thanks{This work is supported by the National Key Research and Development Program of China 2017YFB0502703, the National Natural Science Foundation of China under Grant 61671015, China Scholarship Council and Royal Society International Exchanges 2017 Cost Share (with China). The conference version of this paper was presented at the International Conference on Acoustics, Speech, and Signal Processing (ICASSP), Brighton, UK, May 2019.}
		\thanks{The authors also thank suggestions at ICASSP 2019 on structured total least squares given by Prof. Yoram Bresler from University of Illinois Urbana-Champaign and Dr. Konstantin Usevich from French National Center for Scientific Research. }		
	}
	\title{Dictionary Learning with BLOTLESS Update}
	\maketitle
	\begin{abstract}
		Algorithms for learning a dictionary \textcolor{black}{to sparsely represent a given dataset} typically alternate between sparse coding and dictionary update stages. Methods for dictionary update aim to minimise expansion error by updating dictionary vectors and expansion coefficients given patterns of non-zero coefficients obtained in the sparse coding stage. We propose a block total least squares (BLOTLESS) algorithm for dictionary update. BLOTLESS updates a block of dictionary elements and the corresponding sparse coefficients simultaneously. In the error free case, three necessary conditions for exact recovery are identified. Lower bounds on the number of training data are established so that the necessary conditions hold with high probability. Numerical simulations show that the bounds approximate well the number of training data needed for exact dictionary recovery. Numerical experiments further demonstrate several benefits of dictionary learning with BLOTLESS update compared with state-of-the-art algorithms especially when the amount of training data is small.
	\end{abstract}
	
\section{Introduction \label{sec:Intro}}
Sparse signal representation has found a wide range of applications, including image denoising \cite{elad2006image,liu2017weighted}, image in-painting\cite{elad2006image}, image deconvolution \cite{bronstein2005blind}, image super-resolution \cite{yang2010image,dai2017sparse}, etc. 
The key idea behind the concept of sparse representation is that natural signals tend to have sparse representations under certain bases/dictionaries. 
Hence, finding a dictionary under which a given data set can be represented in a sparse manner, has become a very active area of research. 
Although numerous analytical dictionaries exist, including Fourier basis\cite{Cvetkovic2000On}, discrete cosine transform (DCT) dictionaries, wavelets\cite{Cvetkovic1995Discrete}, curvelets \cite{candes2000curvelets}, etc., the need to adapt to properties of specific data sets has long been driving research efforts towards the development of efficient algorithms for dictionary learning \cite{olshausen1996emergence,aharon2006k}. More formally, dictionary learning is the problem of finding a dictionary $\bm{D}\in \mathbb{R}^{m\times l}$ of $l$ vectors in $\mathbb{R}^{m}$ such that the $n$ training samples in $\bm{Y}\in \mathbb{R}^{m\times n}$ can be written as $\bm{Y}=\bm{D}\bm{X}$, where the coefficient matrix $\bm{X}\in \mathbb{R}^{l\times n}$ is sparse. Of particular interest is overcomplete dictionary learning where the number of dictionary items is larger than the data dimension, i.e., $l>m$, and the number of the training samples is typically much larger than the size of the dictionary,   $n \gg l$. Dictionary learning is a nonconvex bilinear inverse problem, very challenging to solve in general.
	
The bilinear dictionary learning problem is typically approached by alternating between two stages: sparse coding and dictionary update \cite{olshausen1996emergence, aharon2006k, engan1999method, kreutz2003dictionary, dai2012simultaneous, tosic2011dictionary}. 
In the sparse coding stage, the goal is to find sparse representations $\bm{X}$ of training samples $\bm{Y}$ for a given dictionary $\bm{D}$. For that purpose, scores of algorithms have been developed. They can be divided into two main categories. 
The first category consists of greedy algorithms, including orthogonal matching pursuit (OMP) \cite{tropp2007signal}, regularized orthogonal matching pursuit (ROMP) \cite{needell2009uniform}, subspace pursuit (SP) \cite{dai2009subspace}, etc. 
In the second category, sparse coding is formulated as a convex optimization problem where $\ell_1$-norm is used to promote sparsity \cite{chen2001atomic}, and then optimization techniques, e.g. the fast iterative shrinkage-thresholding algorithm (FISTA) \cite{beck2009fast}, can be applied. Reviews of sparse recovery algorithms can be found in \cite{tropp2010computational}. 
	
The goal of the dictionary update  is to refine the dictionary so that the training samples $\bm{Y}$  have more accurate sparse representations given indices of non-zero coefficients obtained in the sparse coding stage. 
In the probabilistic framework, one may apply either maximum likelihood (ML) estimator \cite{olshausen1996emergence} or maximum a posteriori (MAP) estimator \cite{kreutz2003dictionary}, and then solve them by using gradient decent procedures. 
In the context of ML formulation \cite{olshausen1996emergence}, Engan {\sl et al.\ } \cite{engan1999method} proposed the method of optimal directions (MOD) where the sparse coefficients $\bm{X}$ are fixed and the dictionary update problem is cast as a least squares problem which can be solved efficiently;
modifications of MOD were subsequently proposed in \cite{aase2001optimized,skretting2001sparse,engan2007family}.

Recently an alternative approach for dictionary update has become popular, where both the dictionary and the sparse coefficients are updated \emph{simultaneously} with a given sparsity pattern. The representative algorithms include the famous K-SVD algorithm \cite{aharon2006k,smith2013improving} and SimCO \cite{dai2012simultaneous}. 
The crux of K-SVD \cite{aharon2006k} algorithm is to update dictionary items and their corresponding sparse coefficients simultaneously, sequentially one by one. K-SVD was subsequently extended to allow simultaneous update of multiple dictionary elements and corresponding coefficients \cite{smith2013improving}. SimCO \cite{dai2012simultaneous}, of which K-SVD is a special case, goes further and updates the whole dictionary and sparse coefficients simultaneously. The main idea of SimCO is that given a sparsity pattern, the sparse coefficients can be viewed as a function of the dictionary. As a result, the dictionary update becomes a nonconvex optimisation problem with respect to the dictionary. The optimisation is then preformed using the gradient descent method combined with a heuristic sub-routine designed to deal with singular points which can prevent from the convergence to the global minimum\cite{dai2012simultaneous}.
	
Due to the non-convexity of dictionary learning problem, it is challenging to understand under which conditions exact dictionary recovery is possible and which method is optimal in achieving that. Following early efforts on theoretical analysis of exact dictionary recovery \cite{aharon2006uniqueness,hillar2015can,remi2010dictionary,geng2014local,schnass2015local,schnass2014ident,schnass2016conv}, more recently, Spielman et. al. \cite{spielman2012exact} studied dictionary learning problem with complete dictionaries where the dictionary can be presented as a square matrix. By solving a certain sequence of linear programs, they showed that one can recover a complete dictionary $\bm{D}$ from $\bm{Y} = \bm{D}\bm{X}$ when $\bm{X}$ is a sparse random matrix with $O(\sqrt{m})$ nonzeros per column. In \cite{agarwal2014learning,agarwal2013exact,arora2014new,arora2015simple}, the authors propose algorithms which combine clustering, spectral initialization, and local refinement to recover overcomplete  and incoherent dictionaries.

Again these algorithms succeed when $\bm{X}$ has $O(\sqrt{m})$ nonzeros per column. The work in \cite{barak2015dictionary} provides a polynomial-time algorithm that recovers a large class of over-complete dictionaries when $\bm{X}$ has $O(m^{1-\delta})$ nonzeros per column	for any constant $\delta \in (0,1)$. However, the proposed algorithm runs in super-polynomial time when the sparsity level goes up to $O(m)$. Similarly, in \cite{arora2014more} the authors proposed a super-polynomial time algorithm that guarantees recovery with close to $O(m)$ nonzeros per column. Sun et al. \cite{sun2017complete,sun2017complete2}, on the other hand,  proposed a polynomial-time algorithm that provably recovers complete dictionary $\bm{D}$ when $\bm{X}$ has $O(m)$ nonzeros per column and the size of training samples is $O(m^2\log(m))$.

This paper addresses the dictionary update problem, where both the dictionary and the sparse coefficients are updated, for a given sparsity pattern. Whilst it is a sub-problem of the overall dictionary learning, it is nevertheless an important step towards solving the overall problem, and its bilinear nature makes it nonconvex and hence very challenging to solve. Our main contributions are as follows.

\begin{itemize}[itemindent=2em,leftmargin=0pt]
	\item BLOTLESS simultaneously updates a block of dictionary items and the corresponding sparse coefficients. Inspired by ideas presented in \cite{ling2018self,gribonval2012blind}, in BLOTLESS the bilinear nonconvex block update problem is transformed into a linear least squares problem, which can be solved efficiently.
	
	\item For the error-free case, when the sparsity pattern is known exactly, three necessary conditions for unique recovery are identified, expressed in terms of lower bounds on the number of training data. Numerical simulations show that the theoretical bounds well approximate the empirical number of training data needed for exact dictionary recovery. In particular, we show that the number of training samples needed is $O(m)$ for complete dictionary update.
	
	\item BLOTLESS is numerically demonstrated robust to errors in the assumed sparsity pattern.  When embedded into the overall dictionary learning process, it leads to faster convergence rate and less training samples needed compared to state-of-the-art algorithms including MOD, K-SVD and SimCO. 
	
\end{itemize}

Our work is inspired by a recent work \cite{ling2018self} where bilinear inverse problems are formulated as linear inverse problems. The main difference is that our theoretical analysis and algorithm designs in Sections \ref{sec:ExactUpdate} and \ref{sec:Update-Uncertainty} are specifically tailored to the generic dictionary update problem while the focus of \cite{ling2018self} is self-calibration which can be viewed as dictionary learning with only diagonal dictionaries. 
Parts of the results in this paper were presented in the conference paper \cite{Qi2019}. In this journal paper, we refine the bounds in Section \ref{sec:ExactUpdate} and provide detailed proofs, add two total least squares algorithms in Section \ref{sec:Update-Uncertainty}, and include more simulation results in Section \ref{sec:Simulation} to support the design of the algorithm. 

This paper is organized as follows. Section \ref{sec:Background} briefly reviews dictionary learning and update methods. Section \ref{sec:ExactUpdate} discusses an ideal case where exact dictionary recovery is possible, for which a least squares method is developed and analysed. In Section \ref{sec:Update-Uncertainty}, the general case of dictionary update is discussed, and the least squares method is extended to total least squares methods, leading to BLOTLESS. Results of extensive simulations are presented in Section \ref{sec:Simulation} and conclusions are drawn in Section \ref{sec:Conclusion}.

\subsection{Notation}
In this paper, $\| \cdot \|_2$ denotes the $\ell_2$ norm and $\| \cdot \|_F$ stands for the Frobenius norm. For a positive integer $n$, define  $[n]=\{1,2,\cdots,n\}$. For a  matrix $\bm{M}$, $\bm{M}_{i,:}$ and $\bm{M}_{:,j}$ denote the $i$-th row and the $j$-th column of $\bm{M}$ respectively. Consider the sparse coefficient matrix $\bm{X}$. Let $\Omega$ be the support set of $\bm{X}$, i.e., the index set that containing indices of all nonzero entries in $\bm{X}$. Let $\Omega_i$ be the support set of the row vector $\bm{X}_{i,:}$. Then $\bm{X}_{i,\Omega_i}$ is the row vector obtained by keeping the nonzero entries of $\bm{X}_{i,:}$ and removing all its zero entries. Symbols $\bm{I}$, $\bm{1}$, and $\bm{0}$ denote the \textcolor{black}{identity matrix}, the vector of which all the entries are 1, and the vector with all zero entries, respectively. For a given set  $\Omega_i \subset [n]$, $\Omega_i^c$ denotes its complement in $[n]$. 

\section{Dictionary Learning: The Background
\label{sec:Background}}
Dictionary learning is the process of finding a dictionary which sparsely represents given training samples. Let $\bm{Y}\in \mathbb{R}^{m\times n}$ be the training sample matrix, where $m$ is the dimension of training sample vectors and $n$ is the number of training samples. The overall dictionary learning problem is often formulated as:
\begin{equation}
\min_{\bm{D},\bm{X}}\|\bm{Y} - \bm{D}\bm{X}\|^{2}_{F}, \; {\rm s.t.}\; \|\bm{X}_{:,j}\|_{0} \le k,\; \forall j\in [n],	\label{First_DL_model}
\end{equation}
where $\bm{D}\in \mathbb{R}^{m\times l}$ is the dictionary, $\bm{X}\in \mathbb{R}^{l\times n}$ is the sparse coefficient matrix, the $\ell_0$ pseudo-norm $\|\cdot\|_0$ gives the number of non-zero elements, also known as sparsity level, and $k < l$ is the upper bound of the sparsity level.  

Dictionary learning algorithms typically iterate between two stages: sparse coding and dictionary update. The goal of sparsity coding is to find a sparse coefficient matrix $\bm{X}$ for a given dictionary $\bm{D}$. One way to achieve this is to solve the problem 
\begin{equation}
\min_{\bm{X}_{:,j}} \|\bm{Y}_{:,j}-\bm{D}\bm{X}_{:,j}\|_2,\; {\rm s.t.}\; \|\bm{X}_{:,j}\|_{0}\le k,\; \forall j\in [n],
\end{equation}
using greedy algorithms, for example OMP \cite{tropp2007signal} or SP \cite{dai2009subspace}.

In the dictionary update stage, the goal is to refine the dictionary with either fixed sparse coefficients or a fixed sparse pattern, i.e. fixed locations of non-zero coefficients. The famous MOD method \cite{engan1999method} falls into the first category. With fixed sparse coefficients, dictionary update is simply a least squares problem 
\begin{equation*}
\min_{\bm{D}} \|\bm{Y}-\bm{D}\bm{X}\|^2_{F}.
\end{equation*}

A more popular and advantageous approach is to simultaneously update the dictionary and nonzero sparse coefficients by fixing only the sparsity pattern. With this idea, dictionary update is then formulated as  \cite{aharon2006k,dai2012simultaneous,smith2013improving} 
\begin{equation}
\min_{\bm{D},\bm{X}} \|\bm{Y} - \bm{D}\bm{X}\|^2_{F},\ \text{s.t.}\ \mathcal{P}_{\Omega^c}(\bm{X}) = \bm{0}
\label{eq:DictUpdate}
\end{equation}
where $\mathcal{P}_{\Omega^c}( \bm{X})$ gives the vector formed by the entries of $\bm{X}$ indexed by $\Omega^c$. However, problem \eqref{eq:DictUpdate} is bilinear, nonconvex, and challenging to solve. 

Among many methods for solving \eqref{eq:DictUpdate}, we here briefly review K-SVD \cite{aharon2006k} and SimCO \cite{dai2012simultaneous}. K-SVD algorithm successively updates individual dictionary items $\bm{D}_{:,i}$ and the corresponding sparse coefficients $\bm{X}_{i,:}$ whilst keeping all other dictionary items and coefficients fixed:
\begin{equation}
 \min_{\bm{D}_{:,i}, \; \bm{X}_{i,:}} \; \| \left( \bm{Y} - \bm{D}_{:,\{i\}^c} \bm{X}_{\{i\}^c,:}\right)_{:,\Omega_i} - \bm{D}_{:,i} \bm{X}_{i,\Omega_i}\|^{2}_{F}. 
\end{equation}
The optimal solution can be obtained by taking the largest left and right singular vectors of the matrix $\left( \bm{Y} - \bm{D}_{:,\{i\}^c} \bm{X}_{\{i\}^c,:}\right)_{:,\Omega_i}$.

The idea of SimCO is to formulate the dictionary update problem in \eqref{eq:DictUpdate} as a nonconvex optimisation problem with respect to the overall dictionary, that is
\begin{equation}
\begin{aligned}
\min_{\bm{D}} \;\; \underbrace{\min_{\bm{X}:\; \mathcal{P}_{\Omega^c}(\bm{X}) = \bm{0}} \; \|\bm{Y} - \bm{D}\bm{X}\|^2_{F}}_{f(\bm{D})}.
\end{aligned}
\end{equation}
and then solve it using gradient descent of $\bm{D}$. This leads to an update of all dictionary items and sparse coefficients simultaneously. K-SVD can be viewed as a special case of SimCO where the objective function reads
\begin{equation*}
\begin{aligned}
\min_{\bm{D_{:,i}}} \;\; \min_{\bm{X_{i,:}}:\; \mathcal{P}_{\Omega_i^c}(\bm{X}_{i,:}) = \bm{0}} \; \|\bm{Y} - \bm{D}\bm{X}\|^2_{F}.
\end{aligned}
\end{equation*}

The focus of this paper is a novel solution to Problem (\ref{eq:DictUpdate}). 

\section{Exact Dictionary Recovery \label{sec:ExactUpdate} }
This section focuses on an ideal case that the dictionary can be exactly recovered. We assume that the training samples in $\bm{Y}$ are generated from $\bm{Y}=\bm{D}_0\bm{X}_0$ where $\bm{D}_0$ is a tall or square matrix ($m\geq l$) and the sparsity pattern of $\bm{X}_0$ is given. For compositional convenience, we focus on the case where
$\bm{D}_0$ is a square matrix, $\bm{D}_0 \in \mathbb{R}^{m\times m}$, as the same analysis is valid for a tall dictionary where $m > l$.

With given sparsity pattern denoted by $\Omega$, the dictionary update problem can be formulated as a bilinear inverse problem in which the goal is to find $\bm{D}$ and $\bm{X}$ such that 
\begin{equation}
\bm{Y} = \bm{D}\bm{X} \;\; \rm{and} \;\; \mathcal{P}_{\Omega^c}(\mathbf{X}) = \mathbf{0}.
\label{eq:bilinear}
\end{equation}
The constraint $\bm{Y} = \bm{D}\bm{X}$  is nonconvex. Generally speaking, it is challenging to solve (\ref{eq:bilinear}) and there are no guarantees for the global optimality of the solution. 

\subsubsection{Least Squares Solver}

Suppose that the unknown dictionary matrix $\bm{D}$ is invertible. The nonconvex  problem in (\ref{eq:bilinear}) can be translated into a convex problem by using a strategy similar to that explored in \cite{ling2018self}. Define $\bm{H}=\bm{D}^{-1}$. Then $\bm{H}\bm{Y} = \bm{X}$. The goal is now to find $\bm{H}$ and $\bm{X}$ such that 
\begin{equation}
\bm{H}\bm{Y} = \bm{X} \;\; \rm{and} \;\; \mathcal{P}_{\Omega^c}(\mathbf{X}) = \mathbf{0}~,
\label{eq:linear}
\end{equation}
or equivalently, 
\begin{equation}
\left[\bm{Y}^{T},-\bm{I}_{n\times n}\right]\left[\begin{array}{c}
\bm{H}^{T}\\
\bm{X}^{T}
\end{array}\right]=\bm{0}_{n\times m}\;\;{\rm and}\;\;\mathcal{P}_{\Omega^{c}}(\bm{X})=\bm{0},
\label{eq:lestsquare}
\end{equation}
where the subscripts are used to indicate matrix dimensions. In this manner the original bilinear problem (\ref{eq:bilinear}) is cast as an equivalent linear least squares problem. 

However, the formulation in (\ref{eq:lestsquare}) admits trivial solution $\bm{H}=\bm{0}$ and $\bm{X}=\bm{0}$. In fact, (\ref{eq:lestsquare}) admits at least $m$ linearly independent solutions. 

\begin{prop} \label{pro:m-solutions}
	There are at least $m$ linear independent solutions to the least squares problem in (\ref{eq:lestsquare}). 
\end{prop}


\begin{proof}
	This proposition is proved by construction. Let $\bm{H}_0 = \bm{D}_0^{-1}$. 
	Define matrix $\bm{Z}_{i} \in \mathbb{R}^{\left(m+n\right)\times m}$ by keeping the $i$-th column of the matrix $\left[\bm{H}_{0},\bm{X}_{0}\right]^{T}$ and setting all other columns to zero, that is,  $\left(\bm{Z}_{i}\right)_{:,i} = \left[ (\bm{H}_0)_{i,:} , (\bm{X}_0)_{i,:} \right]^{T}$
	and $\left(\bm{Z}_{i}\right)_{:,j} = \bm{0}$ for all $j\ne i$. 
	From the fact that $(\bm{H}_0)_{i,:} \bm{Y} = (\bm{X}_0)_{i,:}$, it is straightforward to verify that $\bm{Z}_i$, $i\in \left[ m \right]$, is a solution of (\ref{eq:lestsquare}). 
	
	The solutions $\bm{Z}_i$, $i\in \left[ m \right]$, are linearly independent. This can be easily verified by observing that the positions of nonzero elements in $\bm{Z}_i$ and $\bm{Z}_j$, $i \ne j$, are different. 		
\end{proof}


\subsubsection{Necessary Conditions for Unique Recovery
\label{subsub:NecessaryConditionUniqueRecovery} }

We now consider the uniqueness of the solution in more detail and derive necessary conditions for unique recovery. Two ambiguities can be identified in the dictionary update problem in (\ref{eq:lestsquare}). The first is permutation ambiguity. Let $\Omega_i$ and $\Omega_j$ be the support sets (the index set containing indices corresponding to nonzero entries) of the $i$-th and $j$-th row of $\bm{X}_0$. If $\Omega_i=\Omega_j$, then the tuple $(\bm{D}_0 \bm{P}_{i\Leftrightarrow j }, \bm{P}_{i\Leftrightarrow j } \bm{X}_0)$ is a valid solution of (\ref{eq:bilinear}), where $\bm{P}_{i\Leftrightarrow j }$ denotes the permutation matrix generated by permuting the $i$-th and $j$-th row of the identity matrix. On the other hand, there is no permutation ambiguity if $\Omega_i \ne \Omega_j$ for all $i \ne j$. In practice, the given sparsity pattern is typically diverse enough to avoid permutation ambiguity. 


The second is the scaling ambiguity which cannot be avoided. Let $\bm{S}$ be a diagonal matrix with nonzero diagonal elements. It is clear that the tuple $(\bm{D}_0 \bm{S}, \bm{S}^{-1} \bm{X}_0)$ is also a valid solution of (\ref{eq:bilinear}). All solutions of the form $(\bm{D}_0 \bm{S}, \bm{S}^{-1} \bm{X}_0)$ form an equivalent class. The scaling ambiguity can be addressed by ntroducing additional constraints. One option used in \cite{ling2018self} is that the sum of the elements in each row of $\bm{H}$ is one, i.e., $\sum_j \bm{H}_{i,j}=1, \forall i $. With these constraints, one has 
\begin{equation}
\bm{H}[\bm{Y},\bm{1}_{m\times 1}] = [\bm{X}, \bm{1}_{m \times 1}] \;\; \rm{and} \;\; \mathcal{P}_{\Omega^c}(\mathbf{X}) = \mathbf{0}.
\label{eq:LeastSquare-Scaling}
\end{equation}

Henceforth, we define unique recovery as unique up to the scaling ambiguity.
\begin{defn}[Unique Recovery] \label{def:UniqueRec}
	The  dictionary update problem (\ref{eq:bilinear}) is said to admit a unique solution if all solutions are of the form $\bm{D} = \bm{D}_0\bm{S}$ and $\bm{X} = \bm{S}^{-1}\bm{X}_0$ for some diagonal matrix $\bm{S}$ with nonzero diagonal elements. 
\end{defn} 

In the following, we identified three necessary conditions for unique recovery. 

\begin{prop}\label{pro:NecessaryBd}
	Assume that $\bm{D}_0$ is square and invertible. If the problem (\ref{eq:bilinear}) has unique solution, then it holds that 
	\begin{enumerate}
		\item \label{cond:1} 
		$n \geq n_0 = m+\frac{| \Omega |}{m}-1$.
		\item \label{cond:2}
		For all $i\in [m]$, the support set of the $i$-th row of $\bm{X}_0$, denoted by $\Omega_i$, satisfies $|\Omega_i^c | \ge m-1$.
		\item \label{cond:3}
		For all $i\in [m]$ and all $i^{\prime} \ne i$, $\exists j \in \left[ n \right] $ such that $(\bm{X}_0)_{i,j} = 0$ and $(\bm{X}_0)_{i^{\prime},j}\neq 0$.
	\end{enumerate}
\end{prop}
\begin{proof}
	Necessary condition \ref{cond:1} is proved by using the fact that the solution of (\ref{eq:LeastSquare-Scaling}) is unique only if the number of equations is larger or equal than the number of unknown variables. The number of unknown variables in (\ref{eq:LeastSquare-Scaling}) is $(n+m)m$ while the number of equations in (\ref{eq:LeastSquare-Scaling}) is $(n+1)m + (nm-|\Omega|)$. Elementary calculations lead to the bound $n_0$.  
	
	The proof of the other two necessary conditions is based on the fact that
	\begin{equation*}
	(\bm{H}_0)_{i,:}\bm{Y} = (\bm{X}_0)_{i,:},
	\end{equation*}
    where $\bm{H}_0 = \bm{D}_0^{-1}$. To simplify the notations, we omit the subscript 0 from $\bm{H}_0$, $\bm{D}_0$ and $\bm{X}_0$ in the rest of this proof.
	Divide the sample matrix $\bm{Y}$ into two sub-matrices  $\bm{Y}_{:,\Omega_i}$ and $\bm{Y}_{:,\Omega_i^c}$. Then it holds that
	\begin{equation*} 
	\bm{H}_{i,:}\bm{Y}_{:,\Omega_i^{c}} = \bm{0}^{T},~{\rm and}~\bm{H}_{i,:}\bm{Y}_{:,j}\ne 0~\forall j\in \Omega_i.  
	\label{null}
	\end{equation*}
    ($\bm{H}_{i,:}^T$ is in the null space of $\bm{Y}_{:,\Omega_i^{c}}$.) Hence $\bm{H}_{i,:}$ is unique (up to a scaling factor) if and only if ${\rm span}(\bm{Y}_{:,\Omega_i^{c}})$ has dimension $m-1$. In this case, ${\rm span}(\bm{H}_{i,:}^T)$ is the null space of both ${\rm span}(\bm{Y}_{:,\Omega_i^{c}})$ and ${\rm span}(\bm{D}_{:,\{i\}^c})$. It is concluded that $\bm{H}_{i,:}$ is unique if and only if 
	${\rm span} ( \bm{Y}_{:,\Omega^{c}_{i}} ) 
	= {\rm span} ( \bm{D}_{ :,\left\{ i \right\} ^{c} } ) 
	$. 
	
	Necessary condition \ref{cond:2} follows directly from that ${\rm rank} ( \bm{Y}_{:,\Omega^{c}_{i}} ) = m-1$. 
	
	To  prove the last necessary condition, note first that the fact that ${\rm span} ( \bm{Y}_{:,\Omega^{c}_{i}} ) 
	= {\rm span} ( \bm{D}_{ :,\left\{ i \right\} ^{c} } ) 
	$ implies that each column of $ \bm{D}_{ :,\left\{ i \right\} ^{c} } $ participates in generating some columns of $\bm{Y}_{:,\Omega^{c}_{i}}$. That is, $\forall i^{\prime}\ne i$, $\bm{D}_{:,i^{\prime}}$ participates in generating $\bm{Y}_{:,j}$ for some $j \notin \Omega^{c}_{i}$. Necessary  condition \ref{cond:3} is therefore proved. 
	Note that condition \ref{cond:3} is not sufficient. It does not prevent the following rank deficient case: there exist $i_1^{\prime},i_2^{\prime}\in \left\{i\right\}^c$ such that both $ \bm{D}_{ :,i_1^{\prime} } $ and $ \bm{D}_{ :,i_2^{\prime} } $ only participate in generating a single sample in $\bm{Y}_{:,j}$ for some $j\in\Omega_i^{c}$.
\end{proof}

\subsubsection{Discussions on the Number of Samples \label{subsub:NumberOfSamples}}

We now study the number of samples $n$ needed to ensure that the necessary
conditions for unique recovery, as specified in Proposition \ref{pro:NecessaryBd}, hold with high probability. To that end we use the following probabilistic model: entries of $\bm{D}_0\in \mathbb{R}^{m \times m}$ are independently generated from the Gaussian distribution $\mathcal{N}(0,\frac{1}{m})$, and entries of $\bm{X}_0 \in \mathbb{R}^{m \times n}$ are independently generated from the Bernoulli-Gaussian distribution $BG(\theta)$ with $\theta \in [0,1]$, where Bernoulli-Gaussian distribution is defined as follows. 
\begin{defn}\label{def:BernGaussDict}
	A random variable $X$ is Bernoulli-Gaussian distributed $X\sim BG(\theta)$ with $\theta \in \left[0,1\right]$, if $X = W\cdot C$, where random variables $W$ and $C$ are independent, $W$ is Bernoulli distributed with parameter $\theta$, and $C\sim \mathcal{N}(0,1)$. 
\end{defn}

\begin{rem}
	The Gaussian distribution is not essential. It can be replaced by any continuous distribution.
\end{rem}

\begin{prop}[Number of Samples]\label{prop:NumOfSamples}
	Suppose that $\bm{Y} = \bm{D}_0 \bm{X}_0$ where $\bm{D}_0$ and $\bm{X}_0$ are generated according to the above probability model. 
	Given a constant $\epsilon\in (0,1)$, the $i$-th necessary condition in Proposition \ref{pro:NecessaryBd} holds with probability at least $1-\epsilon$, if $n \ge n_i$, where 
   \begin{align*}
   n_1 & = \frac{m-1}{1-\theta}\left[1-\frac{\ln\epsilon}{4m\left(m-1\right)\left(1-\theta\right)}\right.\\
   & \quad\left.+\sqrt{\left(1-\frac{\ln\epsilon}{4m\left(m-1\right)\left(1-\theta\right)}\right)^{2}-1}\right].
   \end{align*}
   \begin{align*}
   n_2 & =\frac{m-1}{1-\theta}\left[1-\frac{\ln\epsilon-\ln m}{4\left(m-1\right)\left(1-\theta\right)}\right.\\
   & \quad\left.+\sqrt{\left(1-\frac{\ln\epsilon-\ln m}{4\left(m-1\right)\left(1-\theta\right)}\right)^{2}-1}\right].
   \end{align*}
   and
   \begin{equation*}
   n_{3}=\frac{\ln\epsilon-\ln m-\ln\left(m-1\right)}{\ln\left(1-\theta\left(1-\theta\right)\right)}.
   \end{equation*}
	Furthermore, it holds that $n_1 \le n_2$. If $n \ge \max(n_2,n_3)$, then all three necessary conditions in Proposition \ref{pro:NecessaryBd} hold.
\end{prop}

\begin{proof}
	See Appendix A.
\end{proof}

\begin{rem}
	We have the following observations.
	\begin{itemize}
		\item With fixed $\epsilon$ and $\theta$, $n_1$ and $n_2$ scale linearly with $m$ while $n_3$ is proportional to $\ln m$. 
		\item With fixed $m$ and  $\theta$, $n_1$, $n_2$, and $n_3$ increase proportionally to $-\ln \epsilon$. 
		\item With fixed $m$ and  $\epsilon$, when $\theta$ increases from 0 to 1, $n_1$ and $n_2$ increase, while $n_3$ first decreases and then increases. This matches the intuition that when $\theta$ is too small, we need more samples to have enough information to recover the dictionary. On the other hand, when $\theta$ is too large, more samples are needed to generate the orthogonal space of each $\bm{H}^{T}_{i,:}$. This is verified by simulations in Section \ref{sec:Simulation}. 
	\end{itemize}
\end{rem}

The bound $\max(n_2,n_3)$ provides a good estimate of the number of samples needed for unique recovery. By set theory, if event $A$ is a necessary condition for $B$, then $B$ implies $A$, or equivalently, $B \subseteq A$ and $\Pr(B)\le \Pr(A)$. In Proposition \ref{prop:NumOfSamples}, the quantity $1-\epsilon$ is a lower bound for $\Pr(A)$, where these necessary conditions hold. But unfortunately it is neither lower nor upper bound for $\Pr(B)$, where the dictionary can be uniquely recovered. Nevertheless, our simulations show that $\max(n_2,n_3)$ is a good approximation to the number of samples needed to recover the dictionary uniquely with probability more than $1-\epsilon$. 

In an asymptotic regime, the bounds can be simplified.

\begin{cor}[Asymptotic Bounds]\label{cor:AsymptBd}
	Consider the same settings as in Proposition \ref{prop:NumOfSamples}. For a given $\theta \in (0,1)$, let $m, n\rightarrow \infty$ with $\frac{n}{m} \rightarrow \bar{n} \in \mathbb{R}^{+}$. If $\bar{n}>\frac{1}{1-\theta}$, then all three necessary conditions in Proposition \ref{pro:NecessaryBd} holds with a probability arbitrary close to 1. 
\end{cor}
This corollary follows from elementary calculations and the fact that $\ln(m)/m \rightarrow 0$ when $m \rightarrow \infty$.

\section{Dictionary Update with Uncertainty \label{sec:Update-Uncertainty}}

While Section \ref{sec:ExactUpdate} studies the ideal case, this section focuses on the general case using the insight from Section \ref{sec:ExactUpdate}. In practice, there may be noise in the training samples $\bm{Y}$, and there may be errors in the assumed sparsity pattern. The exact equality in (\ref{eq:bilinear}) may not hold any longer. Following the idea in Section \ref{sec:ExactUpdate}, total least squares methods are applied to handle the uncertainties. The techniques for non-overcomplete and overcomplete dictionaries are developed in Sections \ref{sub:uncertainty-non-overcomplete} and \ref{sub:uncertainty-overcomplete-dictionary} respectively. 

\subsection{Non-overcomplete Dictionary Update \label{sub:uncertainty-non-overcomplete}}
In the case $m \ge l$, let $\bm{H}=\bm{D}^{\dagger}$ be the pseudo-inverse of $\bm{D}$ and assume that $\bm{H}\bm{D}=\bm{I}_{l \times l}$. Due to the uncertainty, Equation \eqref{eq:LeastSquare-Scaling} becomes approximate, that is,
\begin{equation}
\bm{H} [\bm{Y}, \bm{1}] \approx [\bm{X},\bm{1}] \;\; {\rm and} \;\; \mathcal{P}_{\Omega^c} ( \bm{X} )  \approx \bm{0}.
\label{eq:linear_approx}
\end{equation}

Total least squares is a technique to solve a least squares problem in the form $ \bm{A}\bm{X} \approx \bm{B}$ where errors in both observations $\bm{B}$ and regression models $\bm{A}$ are considered\cite{markovsky2007overview,Rhode2014A}. It targets at minimising the total errors via
\begin{equation}
\underset{\tilde{\bm{A}},\tilde{\bm{B}},\bm{X}}{\min} \; \Vert  [ \bm{A}-\tilde{\bm{A}}, \bm{B}-\tilde{\bm{B}} ] \Vert_{F}^{2}, \; {\rm subject\;to} \; \tilde{\bm{A}}\bm{X}=\tilde{\bm{B}}.
\label{eq:TLS-original}
\end{equation}
The constraint set above is nonconvex and hence \eqref{eq:TLS-original} is a nonconvex optimisation problem. Nevertheless, its global optimal solution can be obtained by using the singular value decomposition (SVD). Set $\bm{Z}=[ \bm{X}^T, -\bm{I} ]^T$. Observe that the constraint in \eqref{eq:TLS-original} implies that $[ \tilde{\bm{A}}, \tilde{\bm{B}} ] \bm{Z} = \bm{0}$. The optimal $\bm{Z}$ can be obtained from the smallest right singular vectors of the matrix $[ \bm{A},\bm{B} ]$, and the optimal $[ \tilde{\bm{A}}, \tilde{\bm{B}} ]$ is a best lower-rank approximation of the matrix $[ \bm{A},\bm{B} ]$. 

The difficulty in applying total least squares directly is due to the additional constraint $\mathcal{P}_{\Omega^c} ( \mathbf{X} ) \approx \mathbf{0}$ in \eqref{eq:linear_approx}. Below we present three possible solutions, where the last one IterTLS excels and is adopted. 


\subsubsection{Structured Total Least Squares (STLS) \label{subsub:STLS}}

Consider having uncertainties in both $\bm{Y}$ and the sparsity pattern. Based on \eqref{eq:linear_approx}, a straightforward total least squares formulation is 
\begin{align} 
\min_{\tilde{\bm{Y}},\tilde{\bm{X}},\bm{H}}\; & \frac{1}{2} \Vert \bm{Y}-\tilde{\bm{Y}} \Vert_{F}^{2} + \frac{1}{2} \Vert \mathcal{P}_{\Omega^{c}}\left(\tilde{\bm{X}}\right) \Vert_{2}^{2} \label{eq:sTLS-original}, \\
{\rm s.t.}\; & \bm{H} [\tilde{\bm{Y}},\bm{1} ]= [\tilde{\bm{X}},\bm{1} ]. \nonumber
\end{align}
To solve the above nonconvex optimisation problem, we follow the approach in \cite{lemmerling2003efficient}. It involves an iterative process where each iteration solves an approximated quadratic optimisation problem which admits a closed-form optimal solution. 

At each iteration, denote the initial estimate of $(\tilde{\bm{Y}},\tilde{\bm{X}},\bm{H} )$ by $(\hat{\bm{Y}}, \hat{\bm{X}}, \hat{\bm{H}})$. Note that the constraint set in \eqref{eq:sTLS-original} can be written as $\mathcal{L} (\tilde{\bm{Y}},\tilde{\bm{X}},\bm{H} ) = \bm{0}$ where  
\[
\mathcal{L} (\tilde{\bm{Y}},\tilde{\bm{X}},\bm{H} ) :=\bm{H} [\tilde{\bm{Y}},\bm{1}] - [\tilde{\bm{X}},\bm{1}].
\]
We consider the first order Taylor approximation of $\mathcal{L} (\tilde{\bm{Y}},\tilde{\bm{X}},\bm{H} )$ at given $(\hat{\bm{Y}}, \hat{\bm{X}}, \hat{\bm{H}})$, which reads 
\[
\mathcal{L} ( \tilde{\bm{Y}}, \tilde{\bm{X}}, \bm{H} ) = \mathcal{L} ( \hat{\bm{Y}}, \hat{\bm{X}}, \hat{\bm{H}} ) + \bm{J} ( \bm{z}-\hat{\bm{z}} ),
\]
where $\bm{z}:=[{\rm vect}(\tilde{\bm{Y}})^{T}, {\rm vect}(\tilde{\bm{X}})^{T}, {\rm vect}(\bm{H})^{T}]^{T}$,
$\hat{\bm{z}}:=[ {\rm vect}(\hat{\bm{Y}})^{T}, {\rm vect}(\hat{\bm{X}})^{T}, {\rm vect}(\hat{\bm{H}})^{T}]^{T}$,
and $\bm{J}$ is the corresponding Jacobian matrix. With this approximation, the nonconvex optimisation problem in \eqref{eq:sTLS-original} becomes a quadratic optimisation problem with equality constraints 
\begin{align}
\underset{\tilde{\bm{Y}}, \tilde{\bm{X}}, \bm{H}}{\min} \; & \frac{1}{2} \Vert \bm{Y}-\tilde{\bm{Y}}\Vert_{F}^{2} + \frac{1}{2}\Vert \mathcal{P}_{\Omega^{c}} (\tilde{\bm{X}})\Vert_{2}^{2} \label{eq:STLS-approx},\\
{\rm s.t.}\; & \mathcal{L} ( \hat{\bm{Y}}, \hat{\bm{X}}, \hat{\bm{H}} ) + \bm{J} ( \bm{z} - \hat{\bm{z}} ) = \bm{0}. \nonumber 
\end{align}
This is a quadratic optimisation problem with linear equality constraints, and admits a closed-form solution by a direct application of KKT conditions \cite{Fletcher2013}.

The STLS approach has two issues. The first issue is its very high computational cost. The quadratic optimisation problem \eqref{eq:STLS-approx} involves $mn+ln+lm$ unknowns and $l(n+1)$ equation constraints. Its closed-form solution involves a matrix of size $(n(m+2l)+l(m+1))\times (n(m+2l)+l(m+1))$. We have obtained the closed-form of the Jacobian matrix $\bm{J}$, implemented a conjugate gradient algorithm to use the structures in \eqref{eq:STLS-approx} for a speed-up (details are omitted here). However, simulations in Section \ref{sub:TestTLS} show that the computation speed is still too slow for practical problems. The second issue is the inferior performance compared to other TLS methods in Sections \ref{subsub:ParTLS} and \ref{subsub:ProjectedTLS}. This is because Taylor approximation of the constraint is used in STLS, while other TLS methods below incorporate the constraints directly without Taylor approximation.    

\subsubsection{Parallel Total Least Squares (ParTLS)
\label{subsub:ParTLS}}
The key idea of ParTLS is to decouple the problem \eqref{eq:linear_approx} into $l$ sub-problems that can be solved in parallel: 
\begin{equation*}
\bm{H}_{i,:} [\bm{Y},\bm{1}]\approx [\bm{X}_{i,:},1]
\;{\rm and} \; \bm{X}_{i,\Omega_i^c} \approx \bm{0}, \forall i\in [l].
\end{equation*}
It is straightforward to verify that this is equivalent to 
\begin{equation} \label{eq:decomposition}
\underbrace{\left[\begin{array}{cc} 
	\bm{Y}^{T} & -\bm{P}_{\Omega_{i}}\\
	\bm{1} & \bm{0}
	\end{array}\right]}_{\bm{A}_{i}} \underbrace{\left[\begin{array}{c}
	\bm{H}_{i,:}^{T}\\
	\bm{X}_{i,\Omega_{i}}^{T}
	\end{array}\right]}_{\bm{z}_{i}} \approx \underbrace{\left[\begin{array}{c}
	\bm{0}\\
	1
	\end{array}\right]}_{\bm{b}} , \; 
    \forall i\in\left[l\right],
\end{equation}
where $\bm{P}_{\Omega_i}\in \mathbb{R}^{n \times |\Omega_i|}$ is the projection matrix obtained by keeping the columns of the identity matrix indexed by $\Omega_i$ and removing all other columns. 

Sub-problems (\ref{eq:decomposition}) can be solved by directly applying the TLS formulation \eqref{eq:TLS-original}. Note that $[\tilde{\bm{A}}_{i}, -\tilde{\bm{b}}] [\bm{z}_{i}^{T},1]^{T} = \bm{0}$. The vector $[\bm{z}_{i}^{T},1]^T$ can be computed as a scaled version of the least right singular vector of the matrix $[\bm{A}_{i}, -\bm{b}]$. Then $\bm{H}_{i,:}$ and $\bm{X}_{i,\Omega_i}$ can be obtained from $\bm{z}_i$.

ParTLS enjoys the following advantages. 1) Its global optimality is guaranteed for the ideal case of no data noise or sparsity pattern errors. It is straightforward to see that in the ideal case the ParTLS solutions satisfy \eqref{eq:LeastSquare-Scaling}. 2) It is computationally efficient. The sub-problems (\ref{eq:decomposition}) are of small size and can be solved in parallel. However, ParTLS also has its own issue --- certain structures in the problem are not enforced. For different sub-problem $i\in [l]$, the `denoised' $\bm{Y}$, denoted by $\tilde{\bm{Y}}_i$ can be different.

\subsubsection{Iterative Total Least Squares (IterTLS) \label{subsub:ProjectedTLS}}

IterTLS is an iterative algorithm such that in each iteration a total least squares problem is formulated based on the estimate from the previous iteration.  
It starts with an initial estimate obtained by solving the ideal case equation \eqref{eq:LeastSquare-Scaling}.
In each iteration, let $\hat{\bm{X}}$ be an estimate of $\bm{X}$ from either initialisation or the previous iteration. We formulate the following total least squares problem 
\begin{equation}\label{eq:suboptimalTLS}
\min_{\tilde{\bm{Y}},\tilde{\bm{X}},\bm{H}} \; \left\Vert [\bm{Y}^{T}-\tilde{\bm{Y}}^{T} , \hat{\bm{X}}^{T}-\tilde{\bm{X}}^{T} ] \right\Vert _{F}^{2} \; {\rm s.t.} \; \tilde{\bm{Y}}^{T} \bm{H}^{T} = \tilde{\bm{X}}^{T}, 
\end{equation}
which has the identical form as \eqref{eq:TLS-original}.
Note that the constraint $\mathcal{P}_{\Omega^c}(\bm{X}) \approx \bm{0}$ in \eqref{eq:linear_approx} is implicitly imposed as $\mathcal{P}_{\Omega^c}(\hat{\bm{X}})=\bm{0}$.
The problem \eqref{eq:suboptimalTLS} can be optimally solved by using the SVD
\begin{equation*}
\left[\bm{Y}^{T},\hat{\bm{X}}^{T}\right] = \left[\bm{U}_{Y},\bm{U}_{X}\right] \left[\begin{array}{cc}
\bm{\Sigma}_{Y} & \bm{0}\\
\bm{0} & \bm{\Sigma}_{X}
\end{array}\right] \left[\begin{array}{cc}
\bm{V}_{YY} & \bm{V}_{YX}\\
\bm{V}_{XY} & \bm{V}_{XX}
\end{array}\right]^{T}.
\end{equation*}
The optimal solution is given by $\tilde{\bm{Y}}^{T}  = \bm{U}_{Y} \bm{\Sigma}_{Y} \bm{V}_{YY}^{T}$ and $\tilde{\bm{X}}^{T} = \bm{U}_{Y} \bm{\Sigma}_{Y} \bm{V}_{XY}^{T}$.
To prepare the next iteration, one obtains an updated estimate $\hat{\bm{X}}^{\prime}$ by applying a simple projection operator to $\tilde{\bm{X}}$:  $\mathcal{P}_{\Omega}(\hat{\bm{X}}^{\prime})=\mathcal{P}_{\Omega}(\tilde{\bm{X}})$ and $\mathcal{P}_{\Omega^{c}}(\hat{\bm{X}}^{\prime})=\bm{0}$. With this new estimate $\hat{\bm{X}}^{\prime}$, one can proceed with the next iteration until convergence.


\subsection{Update Overcomplete Dictionary \label{sub:uncertainty-overcomplete-dictionary}}
The difficulty of overcomplete dictionary update comes from the fact that for an overcomplete dictionary $\bm{D}_0$, $\bm{D}_0^{\dagger} \bm{D}_0 \ne \bm{I}$ in general, where $\bm{D}_0^{\dagger}$ is the pseudo-inverse of $\bm{D}_{0}$. Therefore, the above least squares or total least squares approaches cannot be directly applied. 

To address this issue, a straightforward approach is to divide the whole dictionary into a set of sub-dictionaries each of which is either complete or undercomplete, and then update these sub-dictionaries one-by-one whilst fixing all other sub-dictionaries and the corresponding coefficients. More explicitly, given estimated $\bm{D}$ and $\bm{X}$, consider updating $\bm{D}_{:,\mathcal{T}}$,
the submatrix of $\bm{D}$ consisting of columns indexed by
$\mathcal{T}$,
and $\bm{X}_{\mathcal{T,:}}$, 	the submatrix of $\bm{X}$ consisting of rows indexed by $\mathcal{T}$. Then, consider the residual matrix
\begin{equation}
\bm{Y}_r = \bm{Y}-\bm{D}_{:,\mathcal{T}^c} \bm{X}_{\mathcal{T}^c,:},
\end{equation}
and apply the method in Section \ref{subsub:ProjectedTLS} to solve the problem $\bm{Y}_r \approx \bm{D}_{:,\mathcal{T}} \bm{X}_{\mathcal{T},:}$. Then repeat this step for all sub-dictionaries. 
As the dictionary is updated block by block, 
we refer to our algorithm as BLOck Total LEast SquareS (BLOTLESS).

%
%
%

\section{Numerical Test \label{sec:Simulation}}
Parts of the numerical tests are based on synthetic data. The training samples $\bm{Y}=\bm{D}_0 \bm{X}_0$ are generated according to the probability model specified in Section \ref{subsub:NumberOfSamples}. When the dictionary recovery is not  exact, the performance criterion is the difference between the ground-truth dictionary $\bm{D}_0$ and the estimated dictionary $\hat{\bm{D}}$.  In particular, the estimation error is defined as
\begin{equation}
R_{\rm err} = \frac{1}{l}\sum_{p=1}^{l} \left(1-| \bm{\hat{d}}_{p}^{T} \bm{d}_{0_{j_{p}}} |\right),
\label{eq:Dictionary_error}
\end{equation}
where $\bm{\hat{d}}_{p}$ is the $p$-th item in the estimated dictionary, $\bm{d}_{0_{j_{p}}}$ is the $j_p$-th item in the ground-truth dictionary, $j_p = \arg \max_{j\in \mathcal{J}_p} \{|\bm{\hat{d}}_{p}^{T}\bm{d}_{0_j}|\}$ and $\mathcal{J}_p = [l] \setminus \{j_1,j_2,\cdots,j_{p-1}\}$. The items in both dictionaries are normalised to  have unit $\ell_2$-norm. 

Numerical tests based on real data are presented in Section \ref{subsub:RealData} for image denoising. The performance metric is the Peak Signal-to-Noise-Ratio (PSNR) of the denoised images. 

\subsection{Simulations for Exact Dictionary Recovery}
In this section we evaluate numerically bounds in Proposition \ref{prop:NumOfSamples}. All the results presented here are based on 100 random and independent trials. For theoretical performance prediction, we compute $n_2$, $n_3$, and $\max(n_2,n_3)$ using $\epsilon=0.01$. In the numerical simulations, we vary $n$ and find its critical value $n_{\rm sim}$ under which exact recovery happens with an empirical probability at most 99\% and above which exact recovery happens with an empirical probability at least 99\%. 

\begin{figure}[h]
\centering
\subfigure[$\theta=0.1$, $\max(n_2,n_3)=121$]{
	\label{fig:probability:4} 
	\includegraphics[width=0.23\textwidth]{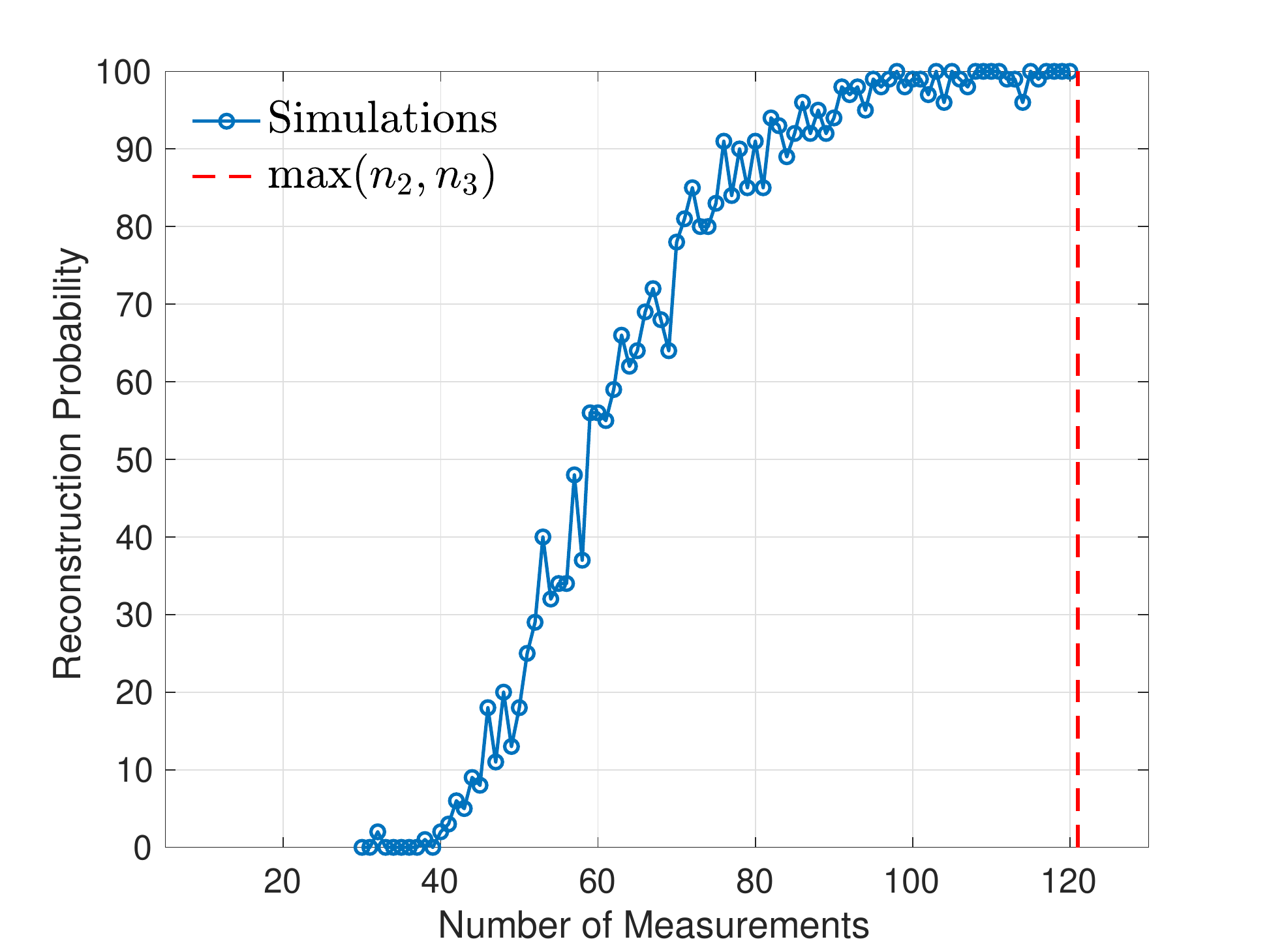}}		
\subfigure[$\theta=0.2$, $\max(n_2,n_3)=65$]{
	\label{fig:probability:6} 
	\includegraphics[width=0.23\textwidth]{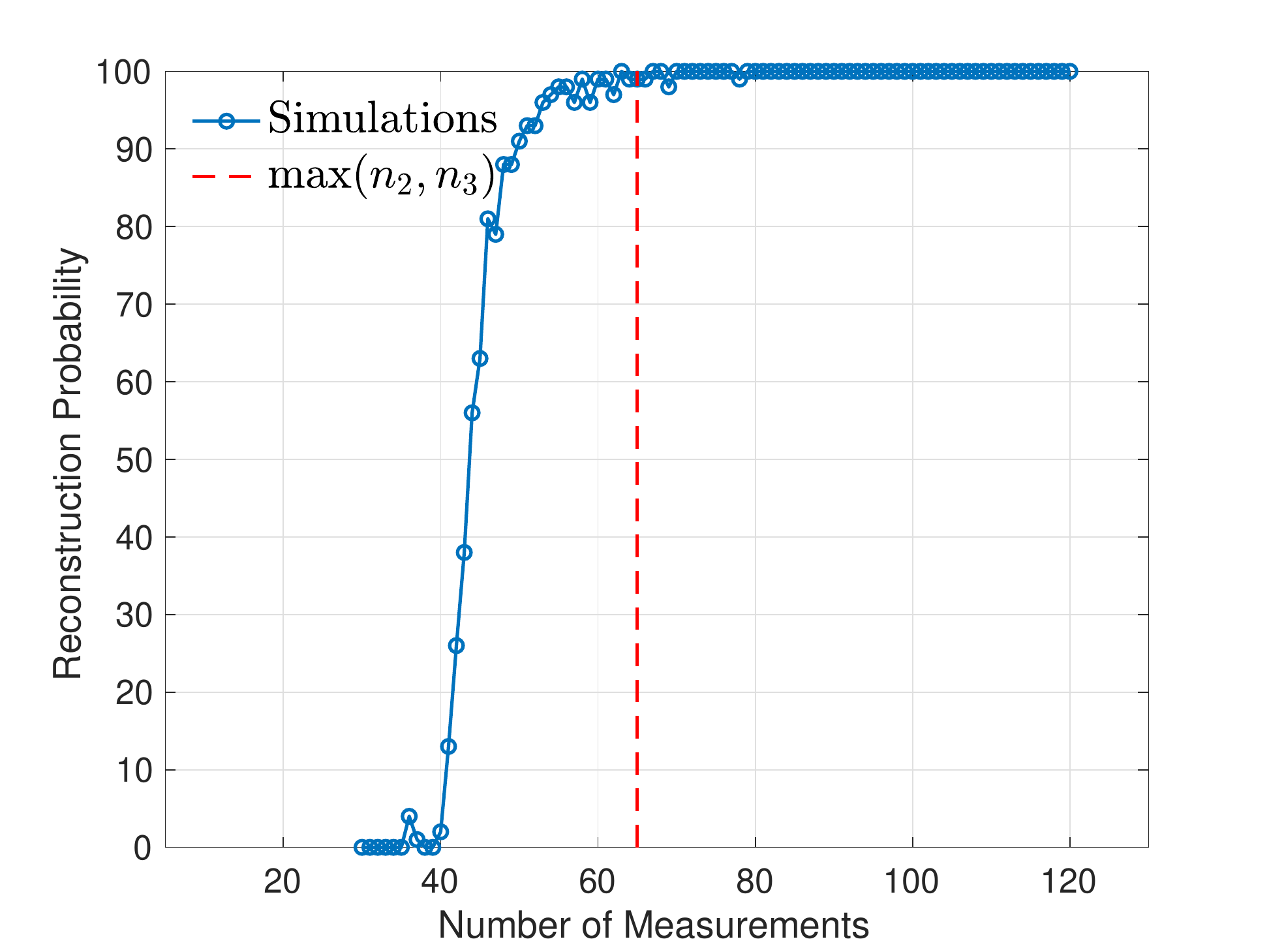}}		
\subfigure[$\theta=0.3$, $\max(n_2,n_3)=64$]{
	\label{fig:probability:10} 
	\includegraphics[width=0.23\textwidth]{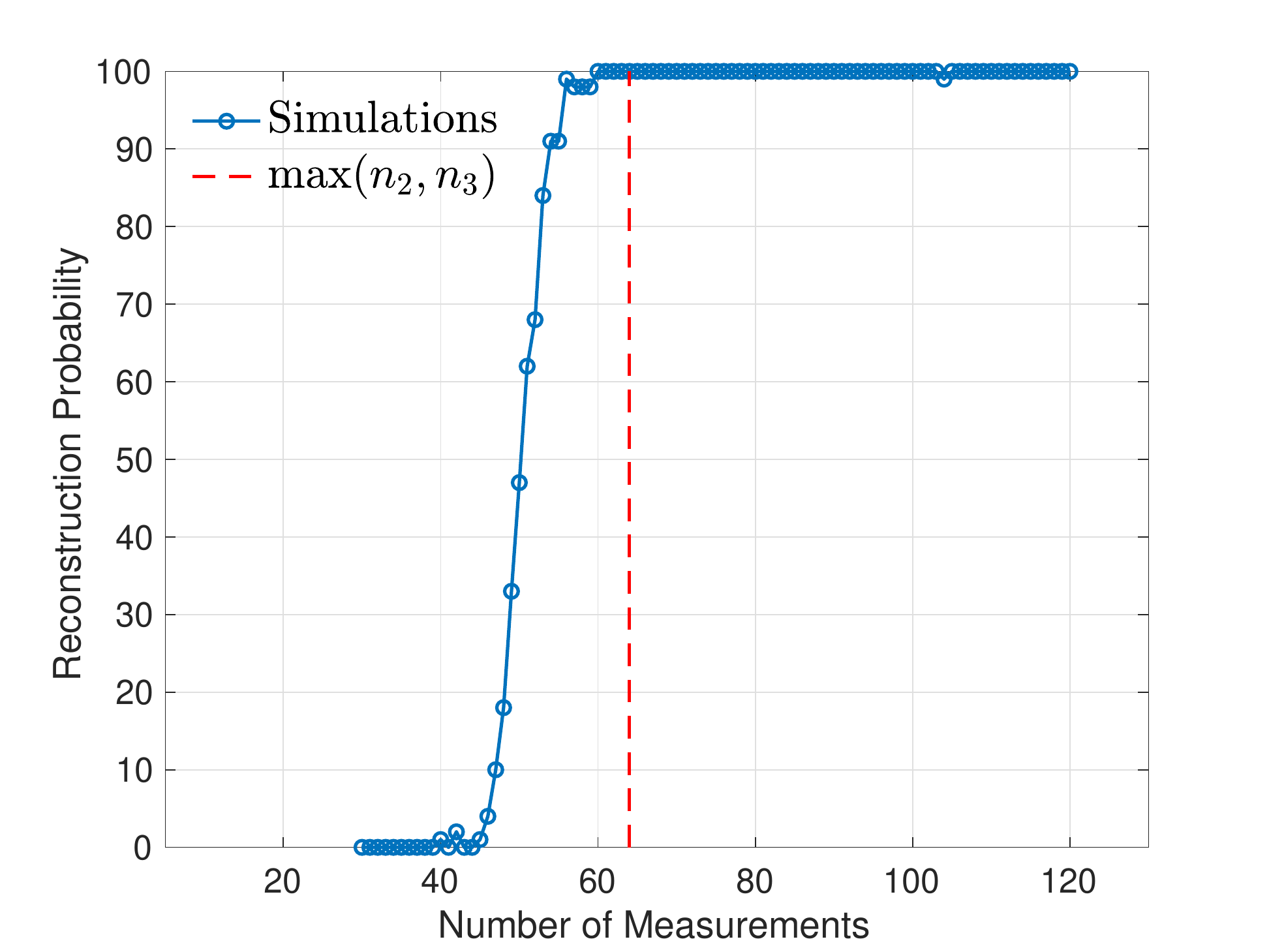}}		 	
\subfigure[$\theta=0.4$, $\max(n_2,n_3)=78$]{
	\label{fig:probability:12} 
	\includegraphics[width=0.23\textwidth]{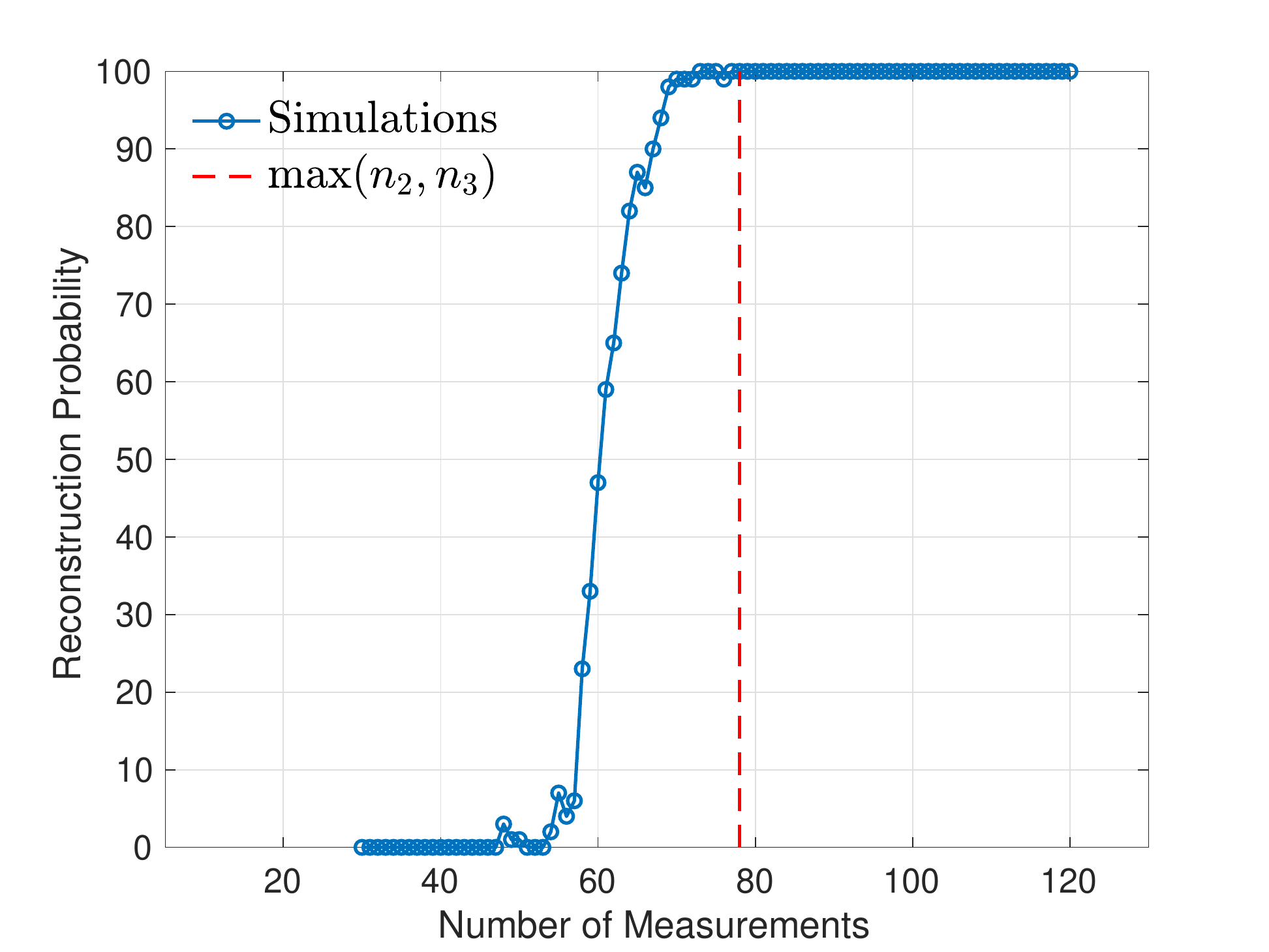}}	
\caption{Probability of exact recovery against the number of training samples for $m=30$. }
\label{fig:probability}		
\end{figure}

\begin{figure}[h]
	\centering
	\subfigure[$m=15$]{
		\label{fig:sim:1} 
		\includegraphics[width=0.23\textwidth]{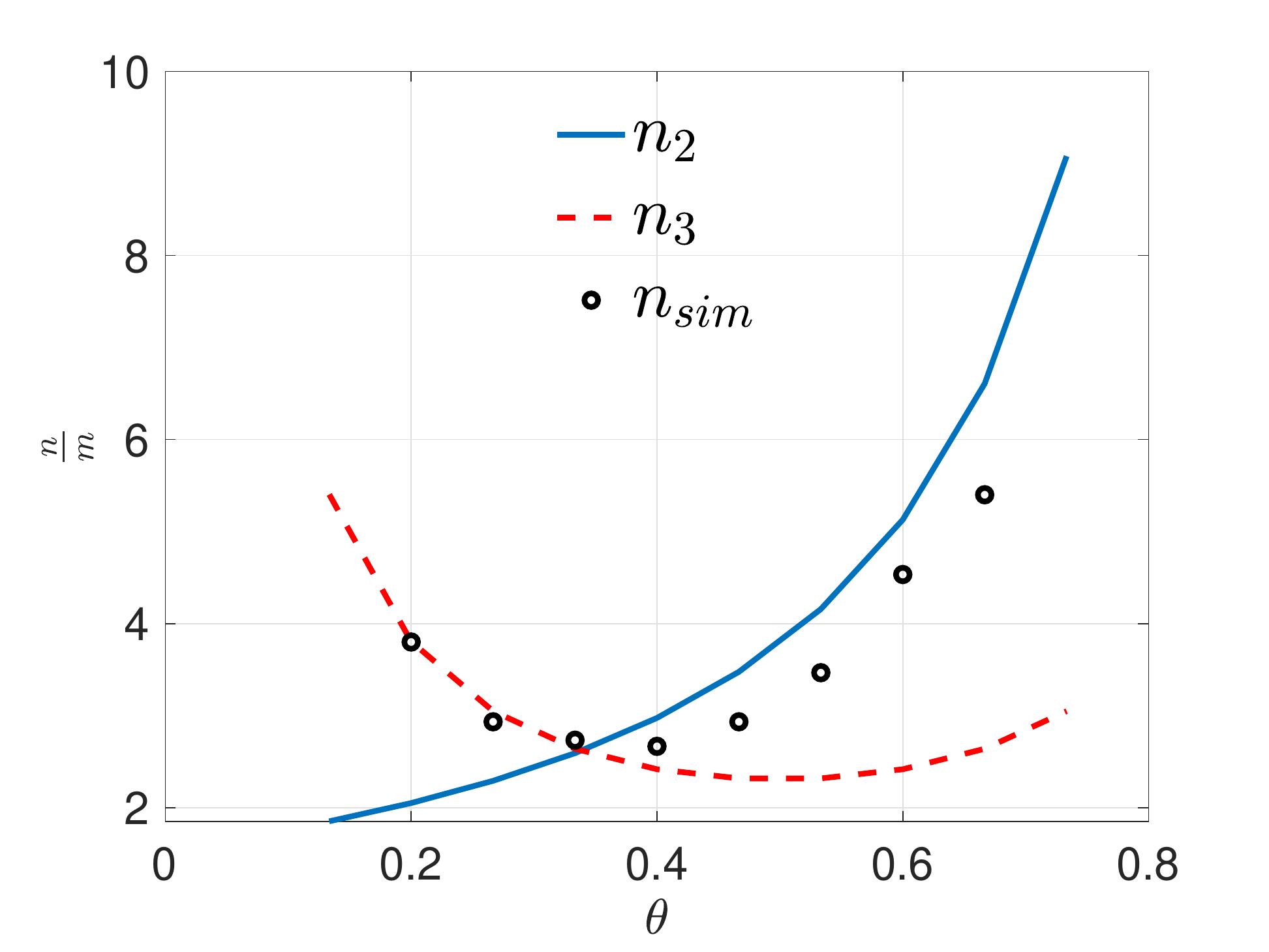}}
	\centering
	\subfigure[$m=20$]{
		\label{fig:sim:2} 
		\includegraphics[width=0.23\textwidth]{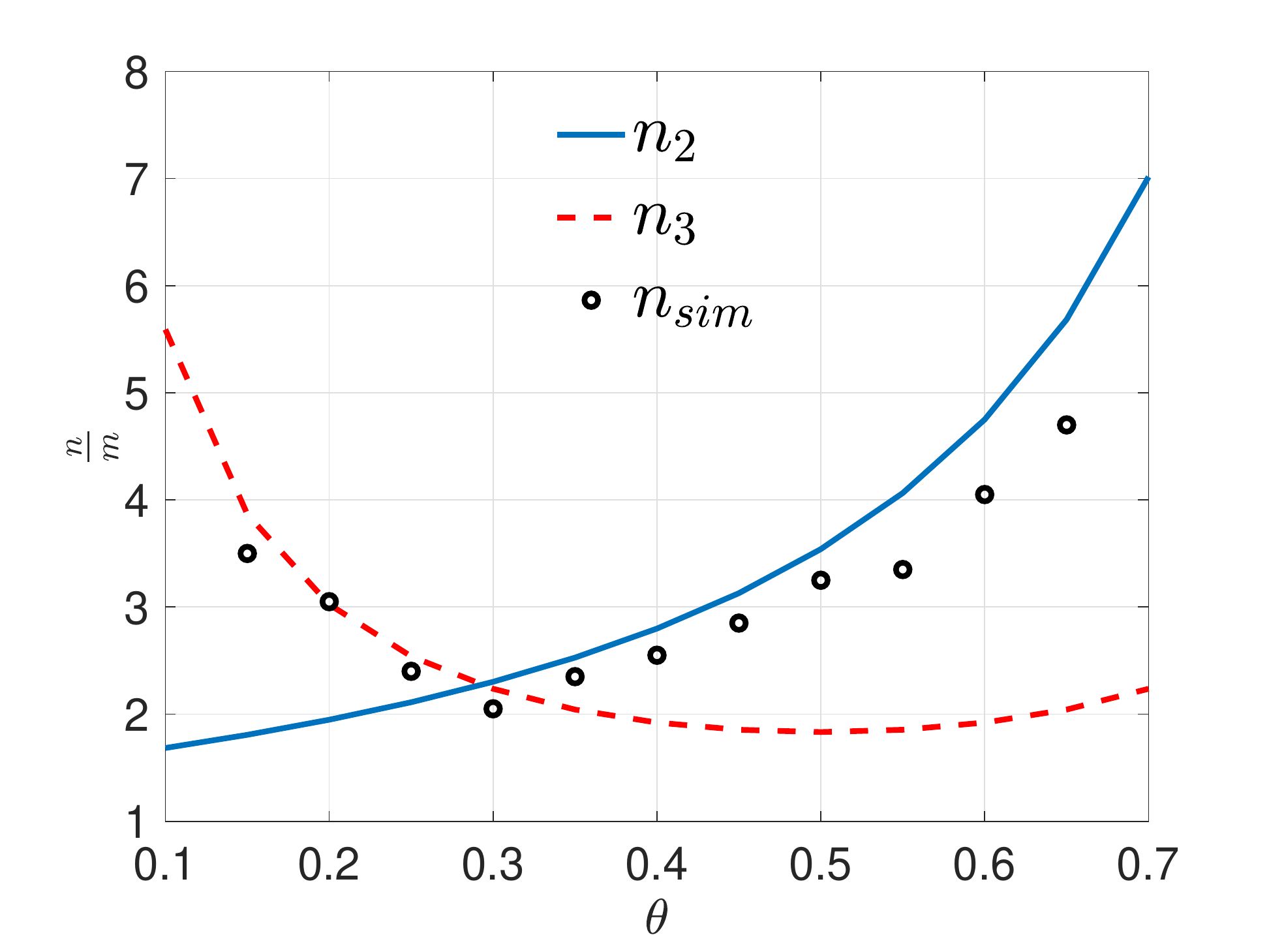}}
	\subfigure[$m=25$]{
		\label{fig:sim:3} 
		\includegraphics[width=0.23\textwidth]{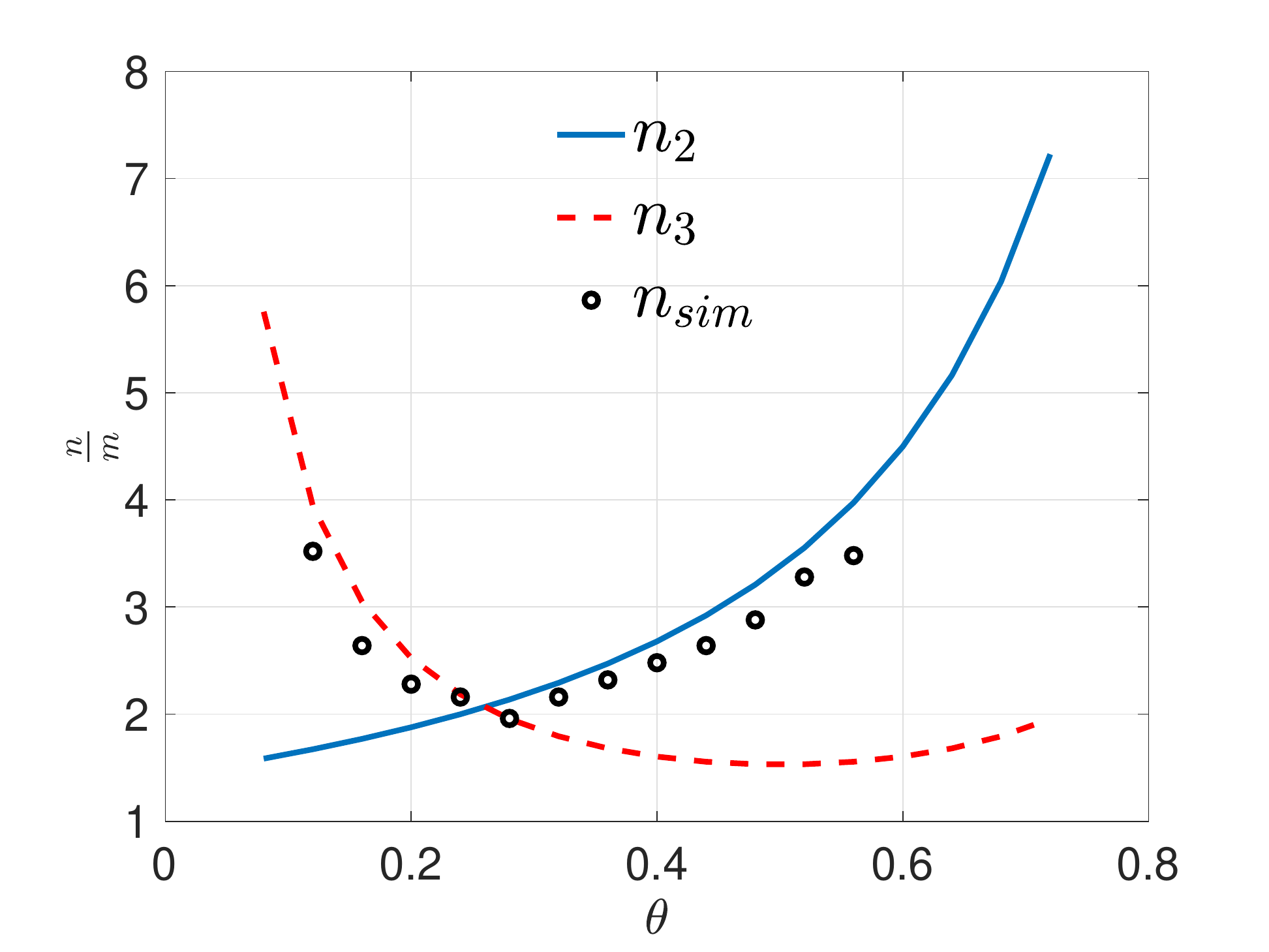}}
	\subfigure[$m=30$]{
		\label{fig:sim:4} 
		\includegraphics[width=0.23\textwidth]{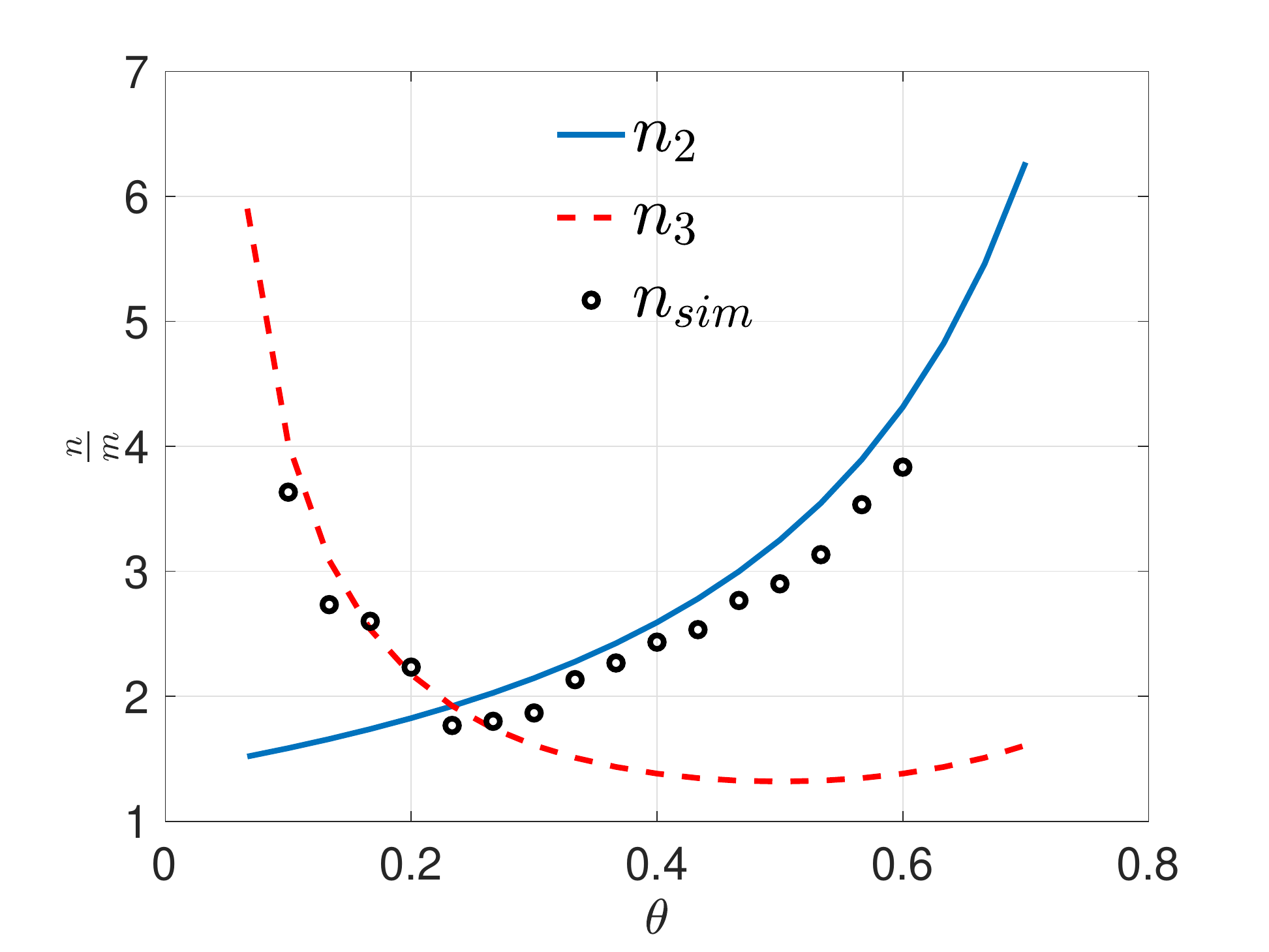}}
	\caption{Normalised number of training samples $n/m$ for at least 99\% probability of exact recovery versus $m$. }
	\label{v_shape}
\end{figure}

We start with the relation between the number of training samples $n$ and the sparsity ratio $\theta$ for a given $m$. In Figure \ref{fig:probability}, we fix $m=30$, vary $\theta$, and study the probability of exact recovery against the number of training samples. Results in Figure \ref{fig:probability} show that the theoretical prediction $\max(n_2,n_3)$ is quite close to the critical $n_{\rm{sim}}$ obtained by simulations. One can also observe that the needed number of training samples first decreases and then increases as $\theta$ increases, which is also predicted by the theoretical bounds. A larger scale numerical test is done in Figure \ref{v_shape}, where four sub-figures correspond to four different values of $m$. Once again, simulations demonstrate these bounds in Proposition \ref{prop:NumOfSamples} match the simulations very well.

\begin{figure}[h]
	\centering
	\includegraphics[width=0.36\textwidth]{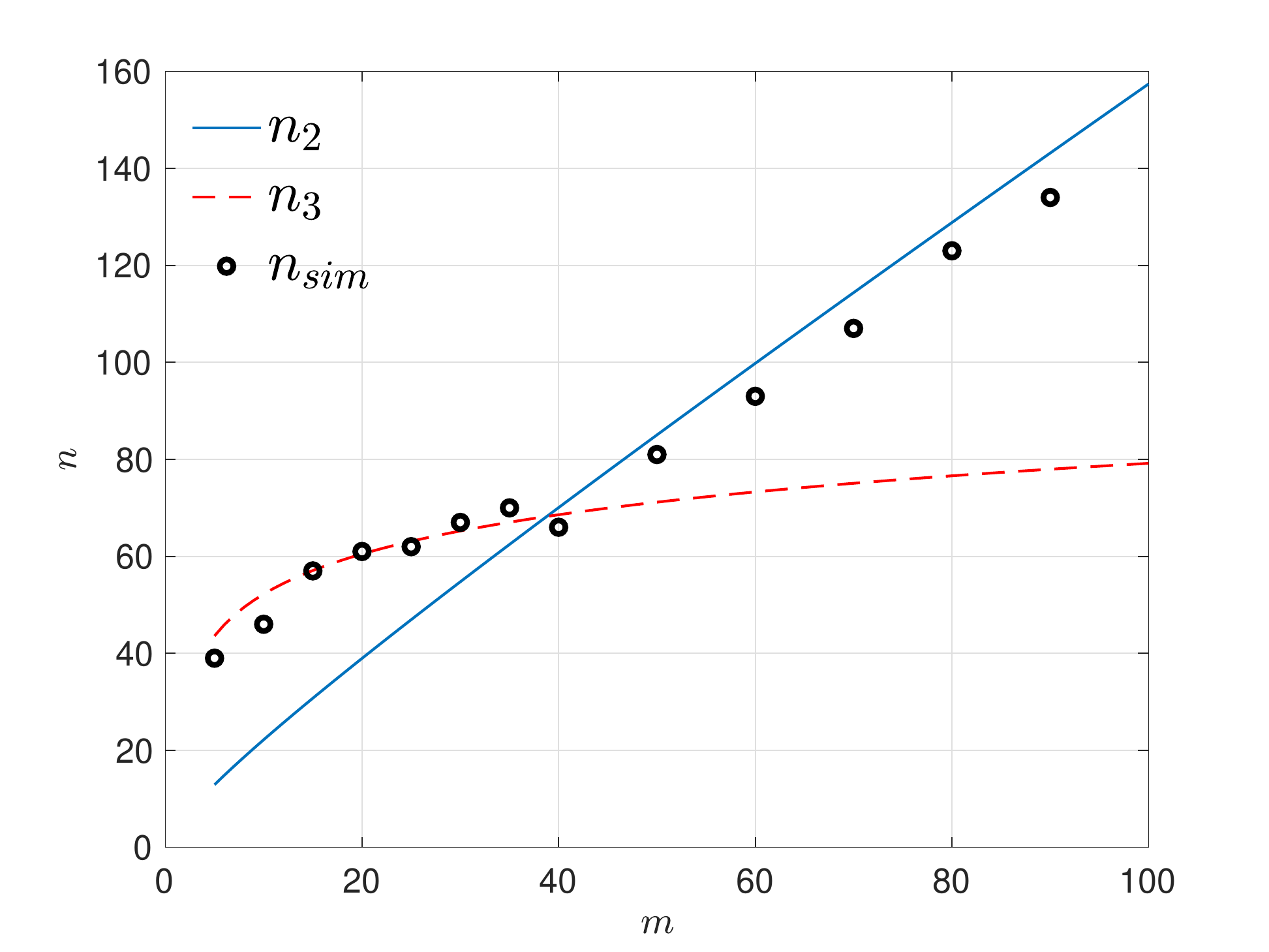}
	\caption{Required $n$ for exact recovery versus $m$ with a given $\theta = 0.2$.}
	\label{sampling_order}
\end{figure}

Let us consider the required $n$ for exact recovery as a function of $m$ by fixing $\theta = 0.2$. Simulation results are depicted in Figure \ref{sampling_order}. When $n_3 \ge n_2$, $n_{\rm sim}$ behaves as $\ln m$. Otherwise, $n_{\rm sim}$ behaves as $m$. This is consistent with Proposition \ref{prop:NumOfSamples}.

\begin{figure}[h]
	\centering
	\includegraphics[width=0.4\textwidth]{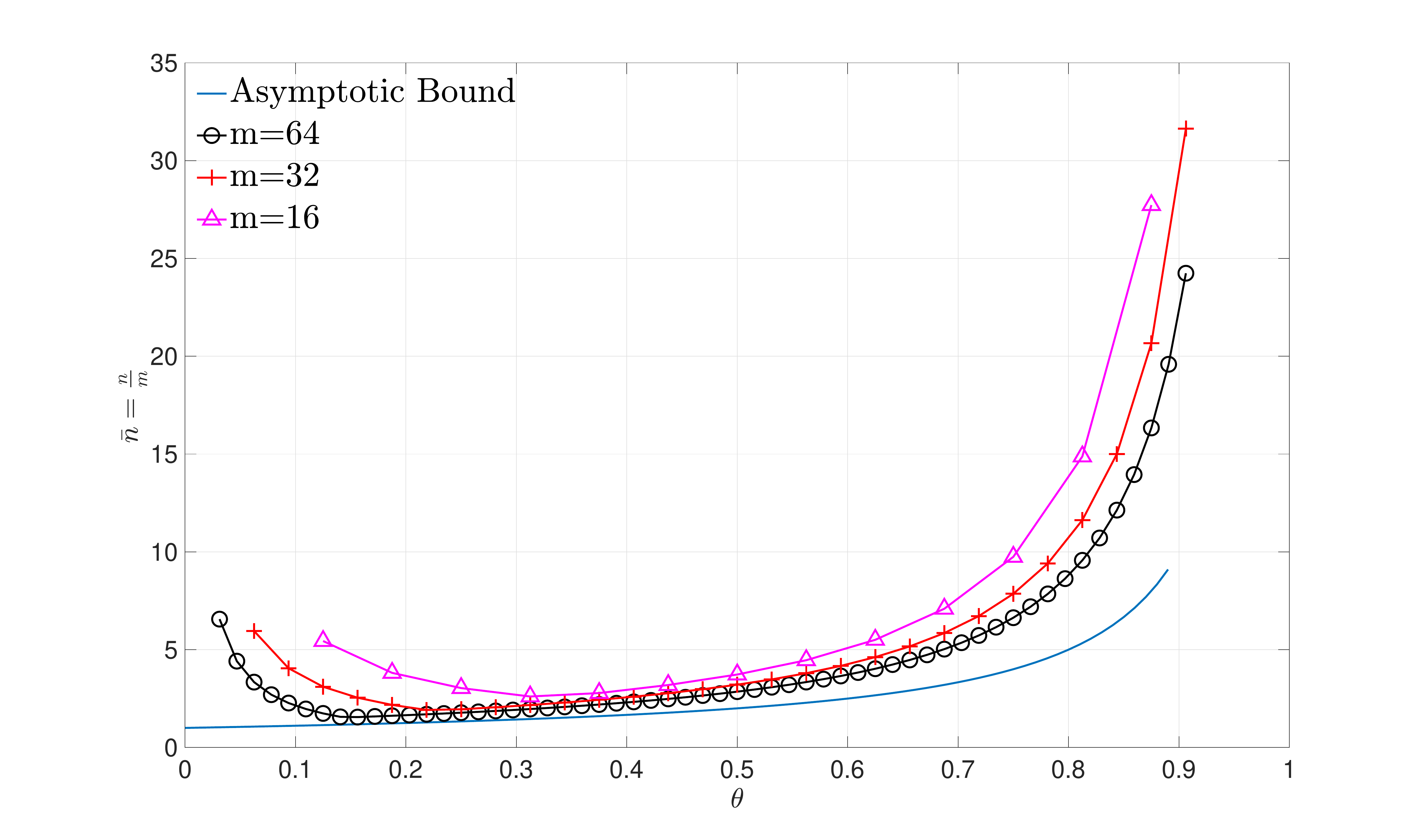}
	\caption{Asymptotic bounds in Corollary \ref{cor:AsymptBd} compared with empirical results for finite $m$. }
	\label{t_bound}
\end{figure}


Finally, we numerically study the accuracy of the asymptotic results in Corollary \ref{cor:AsymptBd}. We draw normalised number of training samples $n/m$ for exact recovery as a function of sparsity ratio $\theta$ for different values of $m$, including $m=16,\;32,\;64$. Simulation results in Figure \ref{t_bound} show a trend that is consistent with the asymptotic results in Corollary \ref{cor:AsymptBd}.

\subsection{Total Least Squares Methods \label{sub:TestTLS}}
The three total least squares methods introduced in Section \ref{sec:Update-Uncertainty} are compared, henceforth denoted by BLOTLESS-STLS, BLOTLESS-ParTLS and BLOTLESS-IterTLS respectively, by being embedded in the dictionary learning process. Random dictionaries are used as the initial starting point of dictionary learning. OMP is used for sparse coding. 




\begin{figure}[h]
	\centering
	\subfigure[$m=l=64$, $\theta=5/64$, $n=400$]{
		\label{fig:complete-d}
		\includegraphics[width=0.23\textwidth]{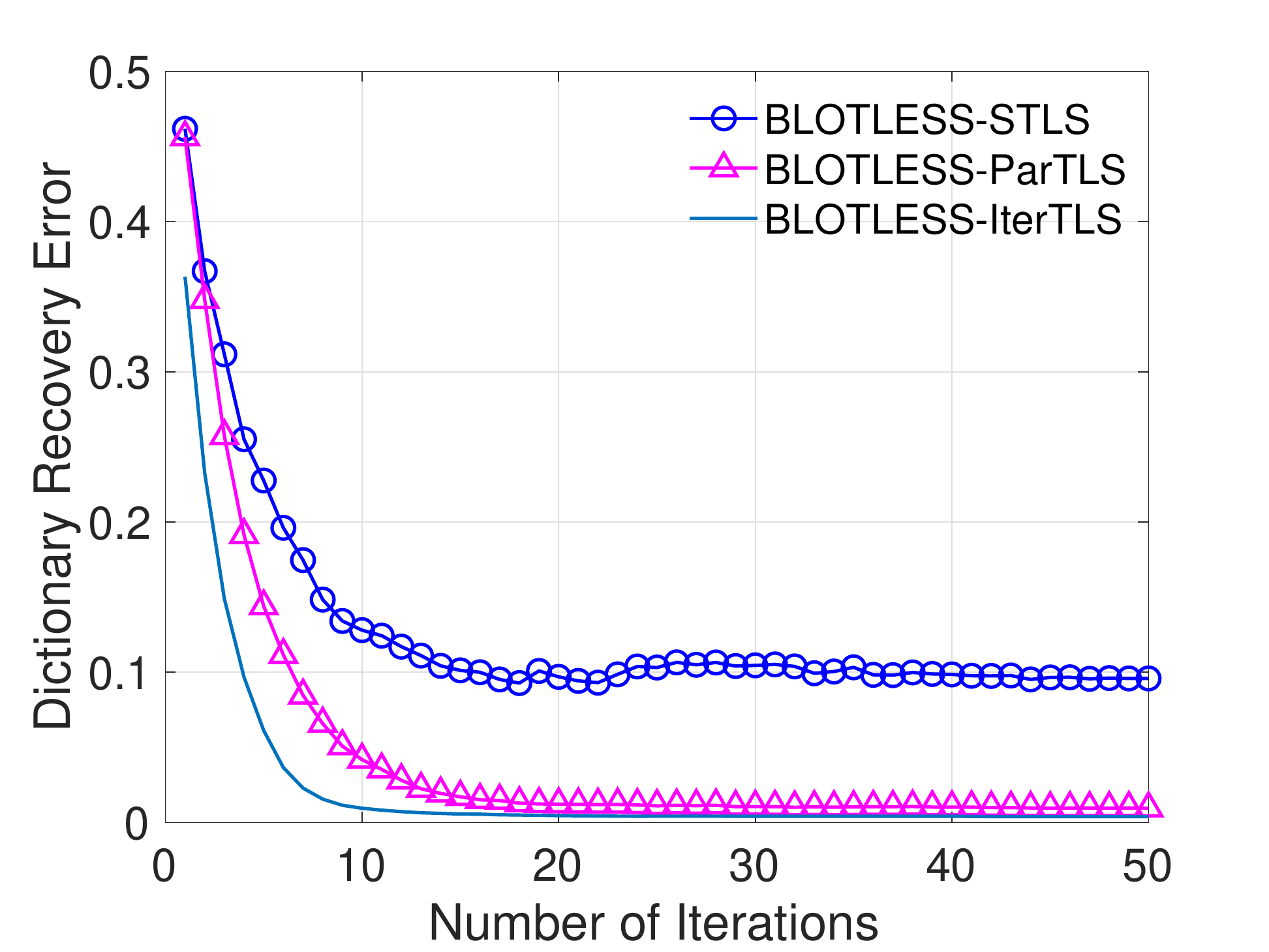}}	
	\hspace{0.025cm}	
	\subfigure[$m=64$, $l=128$, $\theta=5/128$, $n=800$]{
		\label{fig:overcomplete-d} 
		\includegraphics[width=0.23\textwidth]{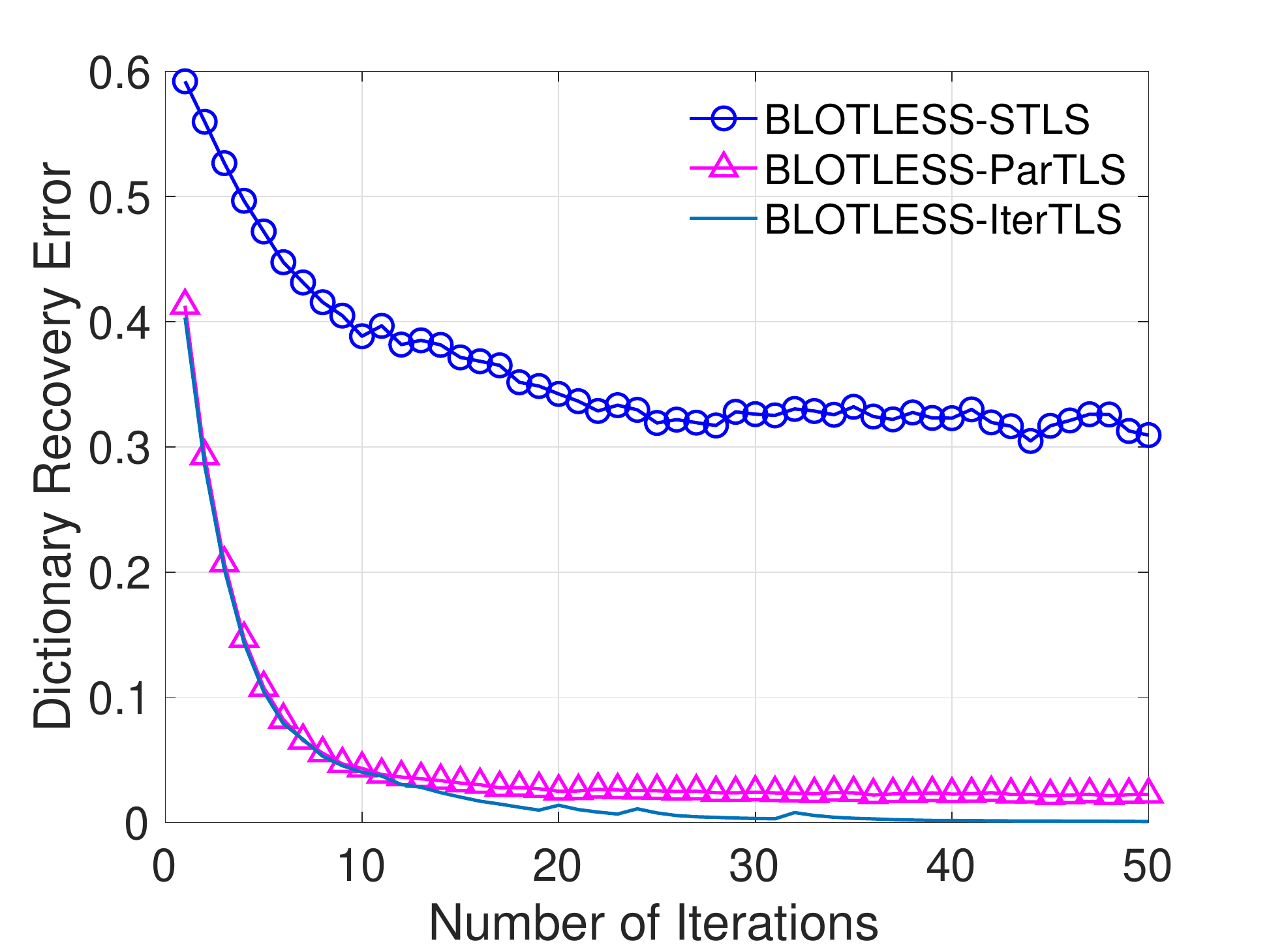}}		
	\caption{Performance comparison of different total least squares methods when used in the overall dictionary learning process. Results are averages of 50 independent trials.}
	\label{fig:comp_STLS_proj}		
\end{figure}

\begin{table}[h]
	\centering
	\caption{Runtime (seconds per iteration of the dictionary learning process) comparison for different total least squares methods: $m=l=64$, $\theta = 5/64$, and number of iterations $n_{\rm it} =50$. Results are averages of 50 independent trials using Matlab 2018b on a MacBook Pro with 16GB RAM and 2.3 GHz Intel Core i5 processor.}
	\label{tab:Run_time}
	\begin{tabular}{c|c|c|c|c}
		\hline
		& $n=200$ &$n=300$ & $n=400$ & $n=500$\\
		\hline
		STLS & 621.6544 & 836.3027 & 1098.5955 & 1265.1895 \\
		\hline
		ParTLS & 9.2544 & 18.1191 & 25.3678 & 31.0551 \\
		\hline
		IterTLS & 6.3446 & 8.5157 & 10.9026 & 13.8056\\
		\hline
	\end{tabular}
\end{table}

Fig. \ref{fig:comp_STLS_proj} compares the dictionary learning errors \eqref{eq:Dictionary_error} for both complete and over-complete dictionaries.
BLOTLESS-IterTLS converges the fastest and has the smallest error floor. 
Then a runtime comparison is given in Table \ref{tab:Run_time}. BLOTLESS-IterTLS  clearly  outperforms the other two methods. It is therefore used as the default dictionary update method in later simulations. 

\subsection{Robustness to Errors in Sparsity Pattern}

\begin{figure}[h]
	\centering
	\includegraphics[width=0.36\textwidth]{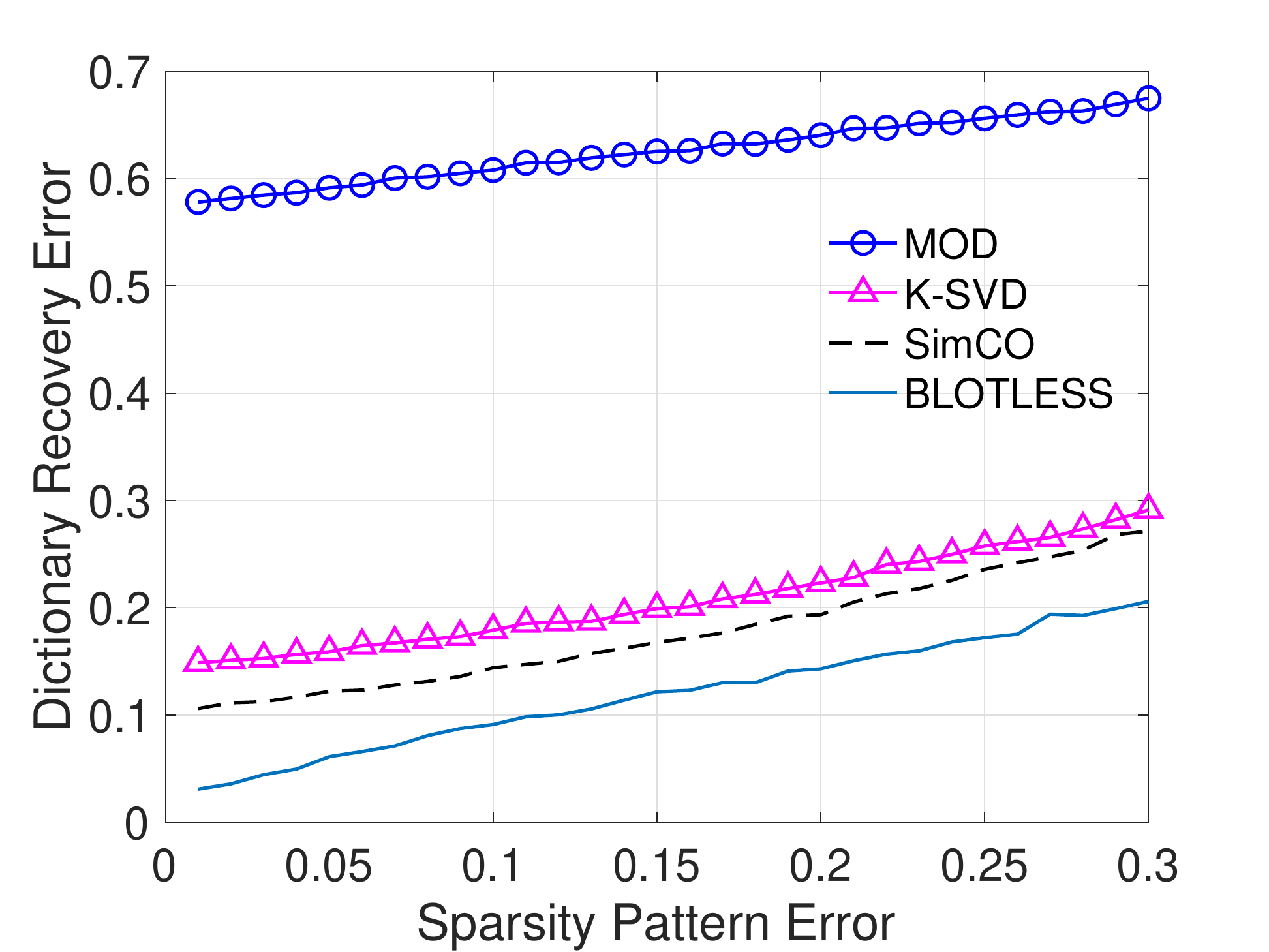}
	\caption{Performance comparison of different update algorithm when different sparse pattern mismatch is applied. $m=l=64$, $\theta=5/64$, $n=500$. Results are averages of 100 independent trials.}
	\label{fig:Sp_micmatch2}
\end{figure}

Simulations are next designed to evaluate the robustness of different dictionary update algorithms to errors in sparsity pattern. Towards that end, a fraction $r$ of indices in the true support are randomly chosen to be replaced with the same number of randomly chosen indices not in the true support set. This erroneous sparsity pattern is then fed into different dictionary update algorithms. The numerical results in Fig. \ref{fig:Sp_micmatch2} demonstrate the robustness of BLOTLESS (with IterTLS). 

\subsection{Dictionary Learning with BLOTLESS Update}
This subsection compares dictionary learning performance for different dictionary update methods. The sparse coding algorithm is OMP. IterTLS in Section \ref{subsub:ProjectedTLS} is used for BLOTLESS. Results for synthetic data are presented in Section \ref{subsub:SyntheticData} while Section \ref{subsub:RealData} focuses on image denoising using real data.

\subsubsection{Synthetic Data \label{subsub:SyntheticData}}

\begin{figure}[h]
	\centering
	\subfigure[$m=l=64$, $n=400$, $\theta=5/64$.]{
		\label{fig:3} 
		\includegraphics[width=0.23\textwidth]{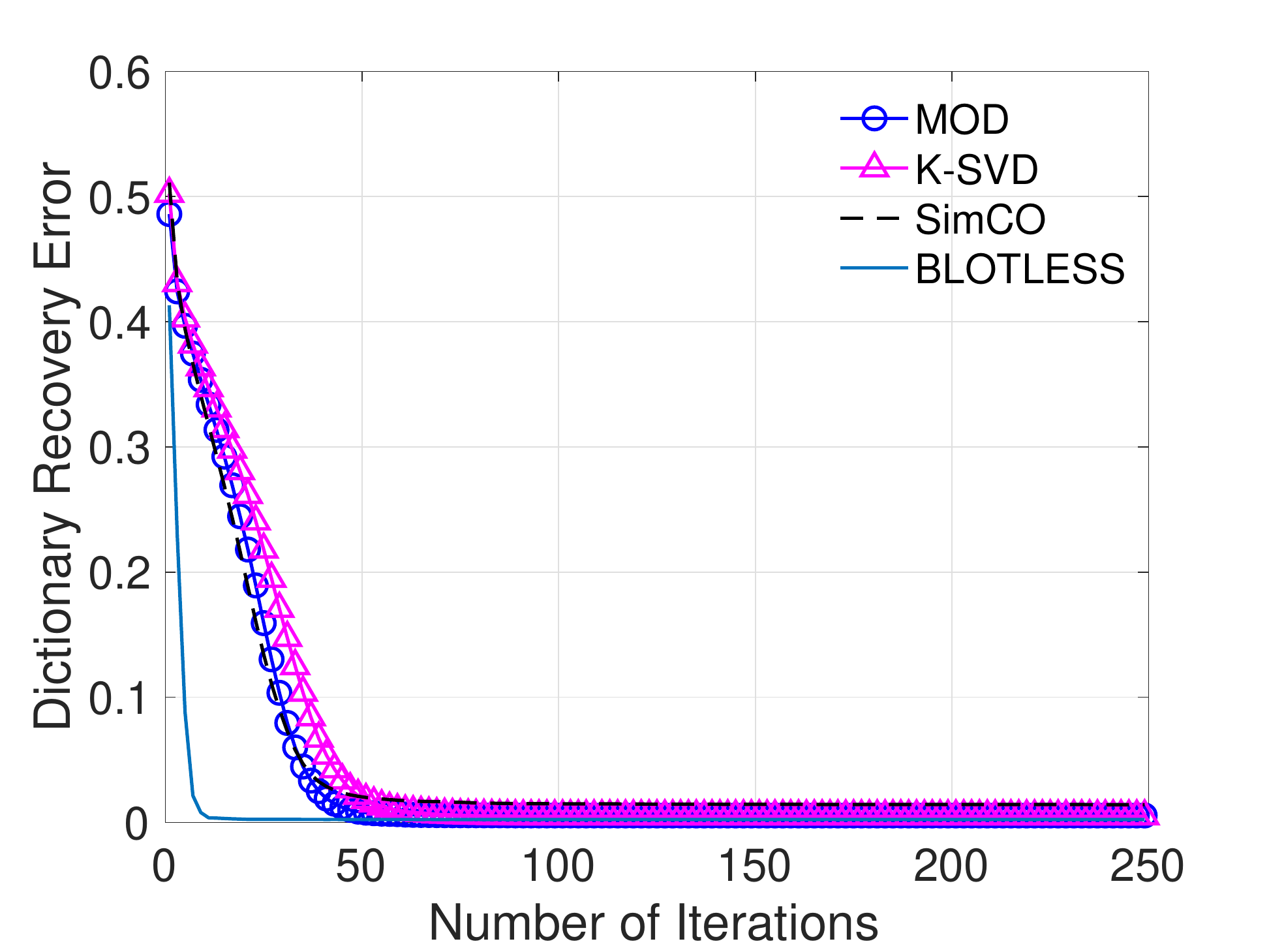}}
	\hspace{0.025cm}
	\subfigure[$m=64$, $l=128$, $n=600$, $\theta=5/128$.]{
		\label{fig:4} 
		\includegraphics[width=0.23\textwidth]{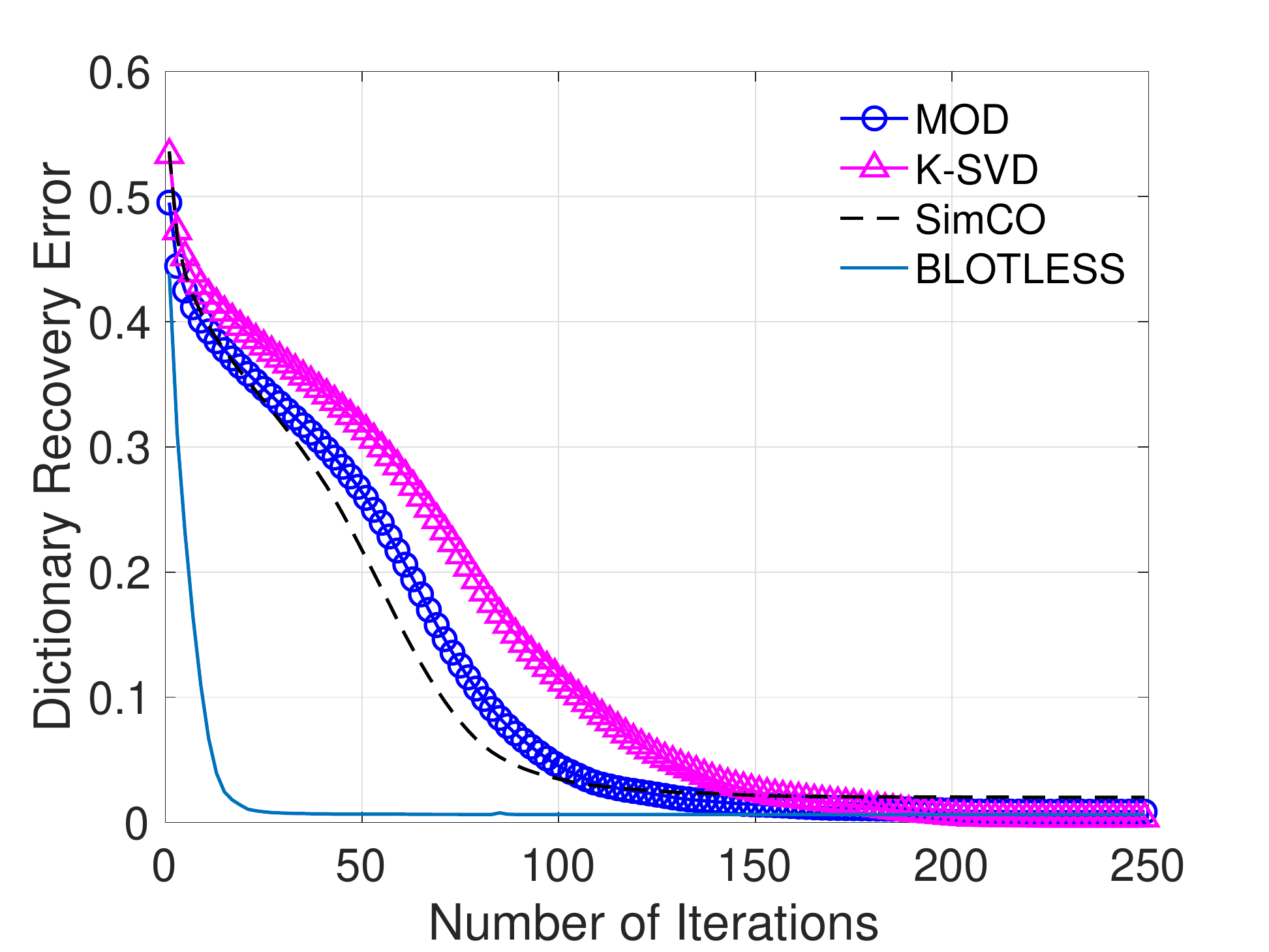}}
	\subfigure[$m=l=64$, $\theta=5/64$, $n_{\rm it}=150$.]{
		\label{fig:5} 
		\includegraphics[width=0.23\textwidth]{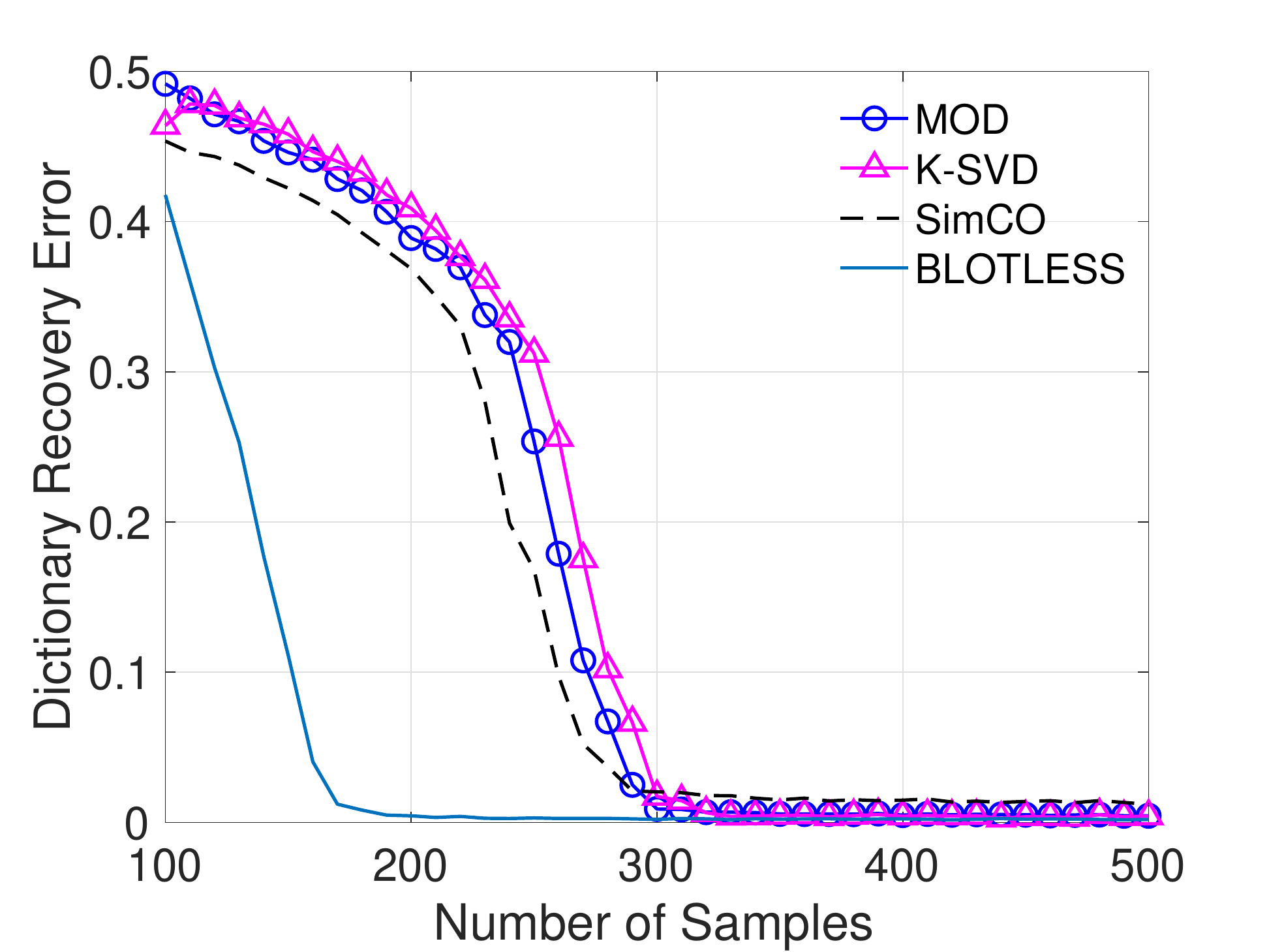}}
	\hspace{0.025cm}
	\subfigure[$m=64$, $l=128$, $\theta=5/128$, $n_{\rm it}=150$.]{
		\label{fig:7} 
		\includegraphics[width=0.23\textwidth]{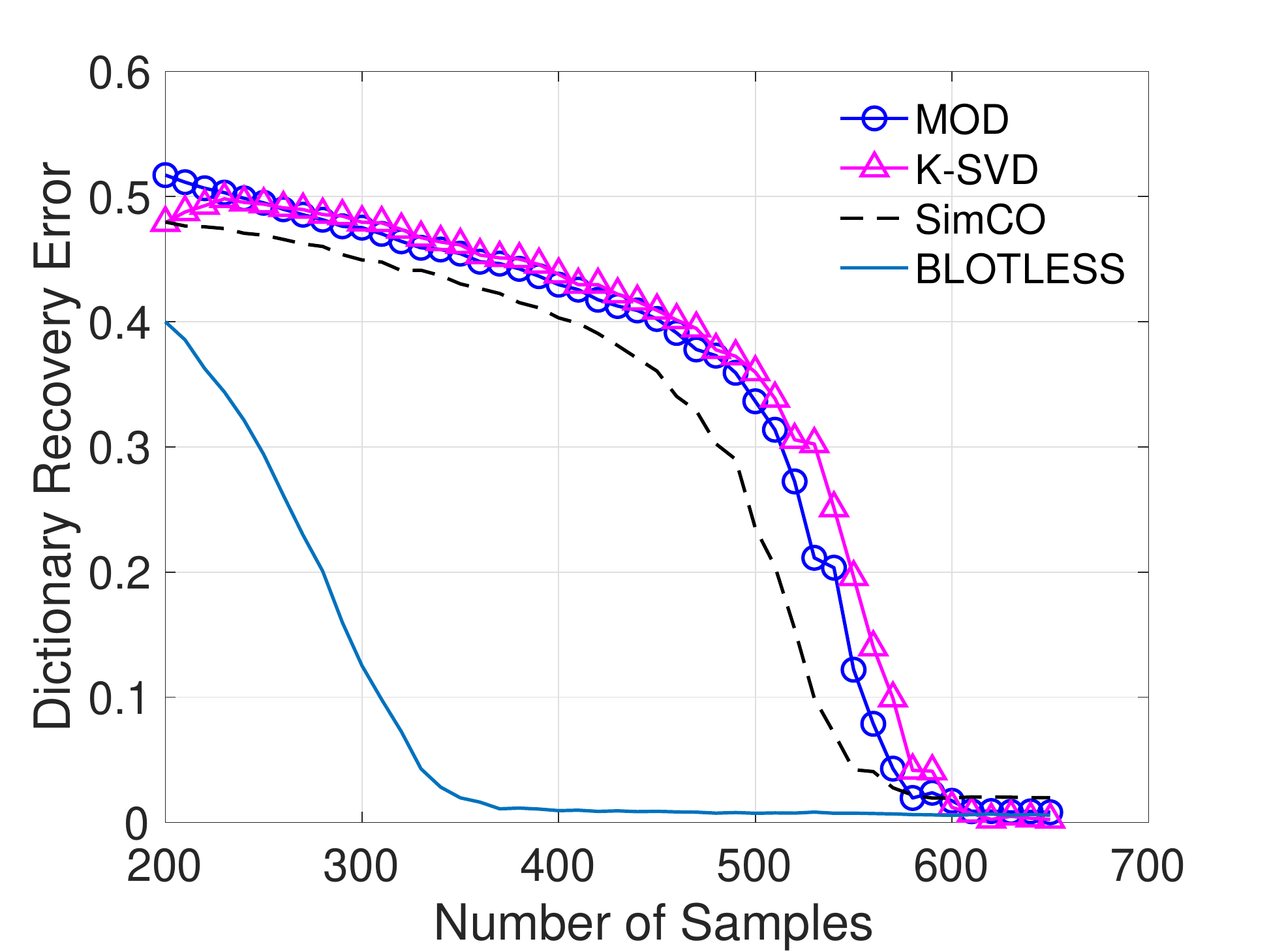}}
	\caption{Comparison of dictionary update methods for the noise free case. Results are averages of 100 trials.}
	\label{fig:DL_noisefree}
\end{figure}

\begin{figure}[h]
	\centering
	\subfigure[$m=l=64$, $n=500$, $\theta=5/64$.]{
		\label{fig:9} 
		\includegraphics[width=0.23\textwidth]{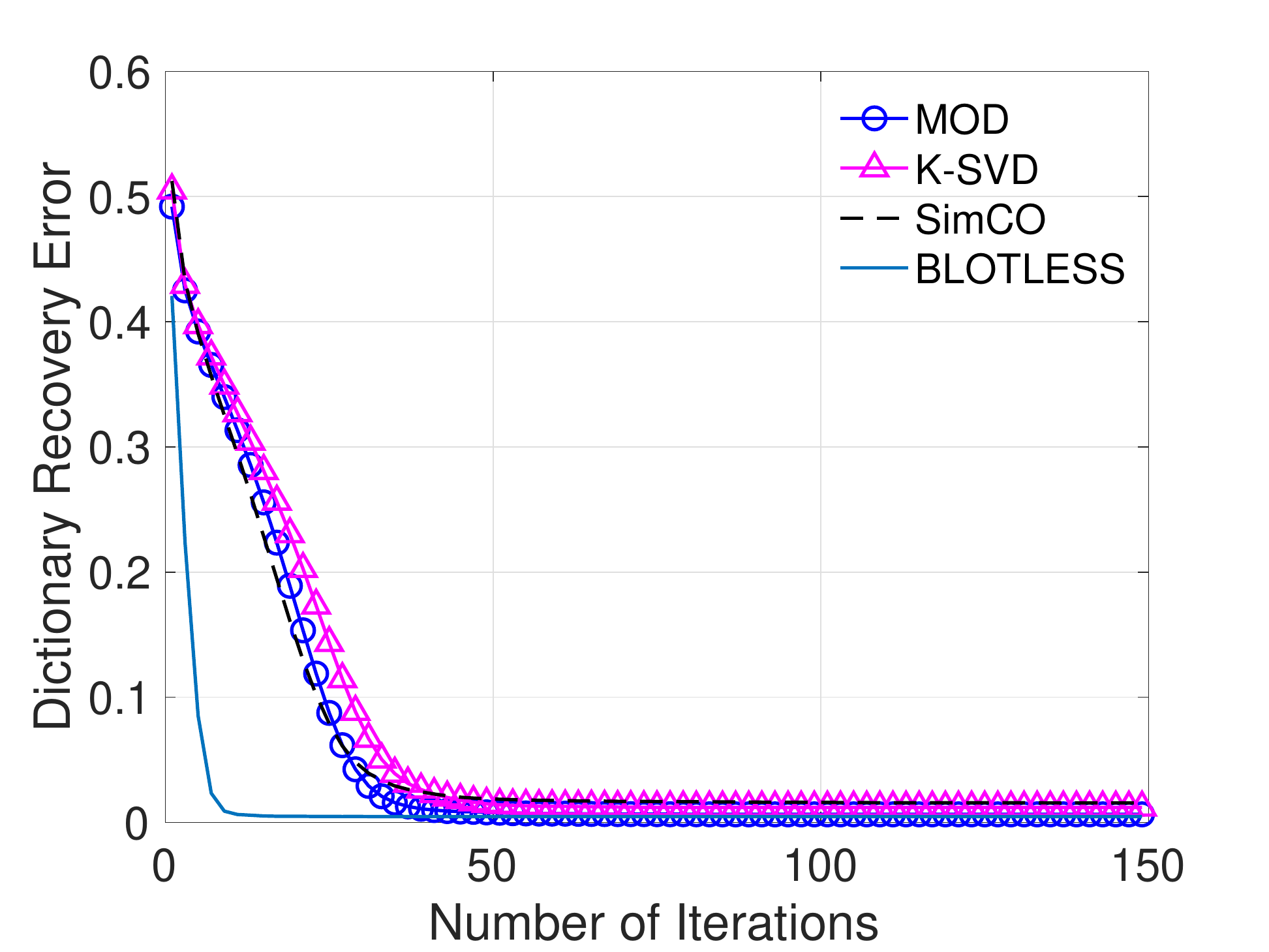}}
	\hspace{0.025cm}
	\subfigure[$m=64$, $l=128$, $n=700$, $\theta=5/128$.]{
		\label{fig:10} 
		\includegraphics[width=0.23\textwidth]{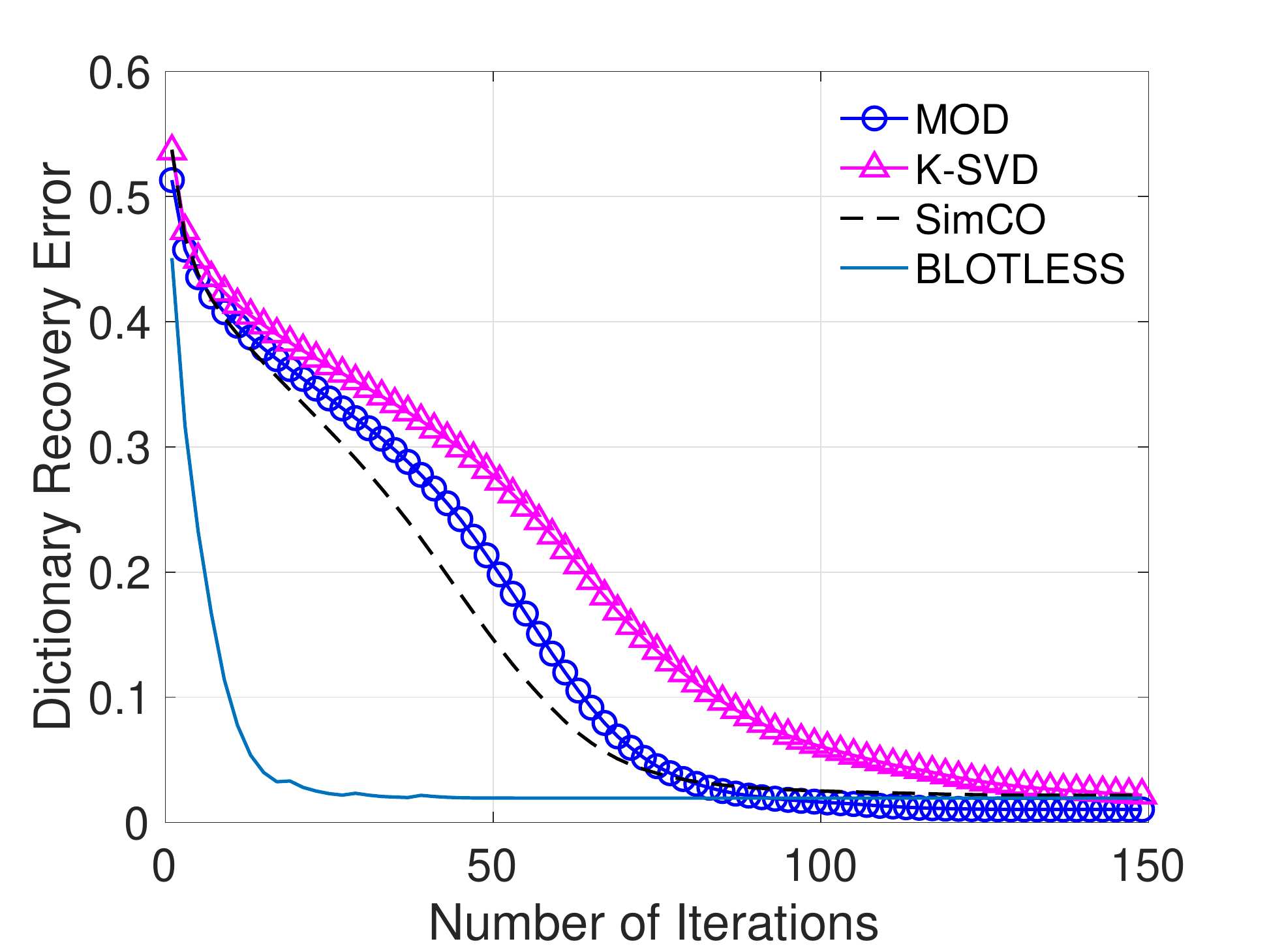}}
	\subfigure[$m=l=64$, $\theta=5/64$, $n_{\rm it}=150$.]{
		\label{fig:8} 
		\includegraphics[width=0.23\textwidth]{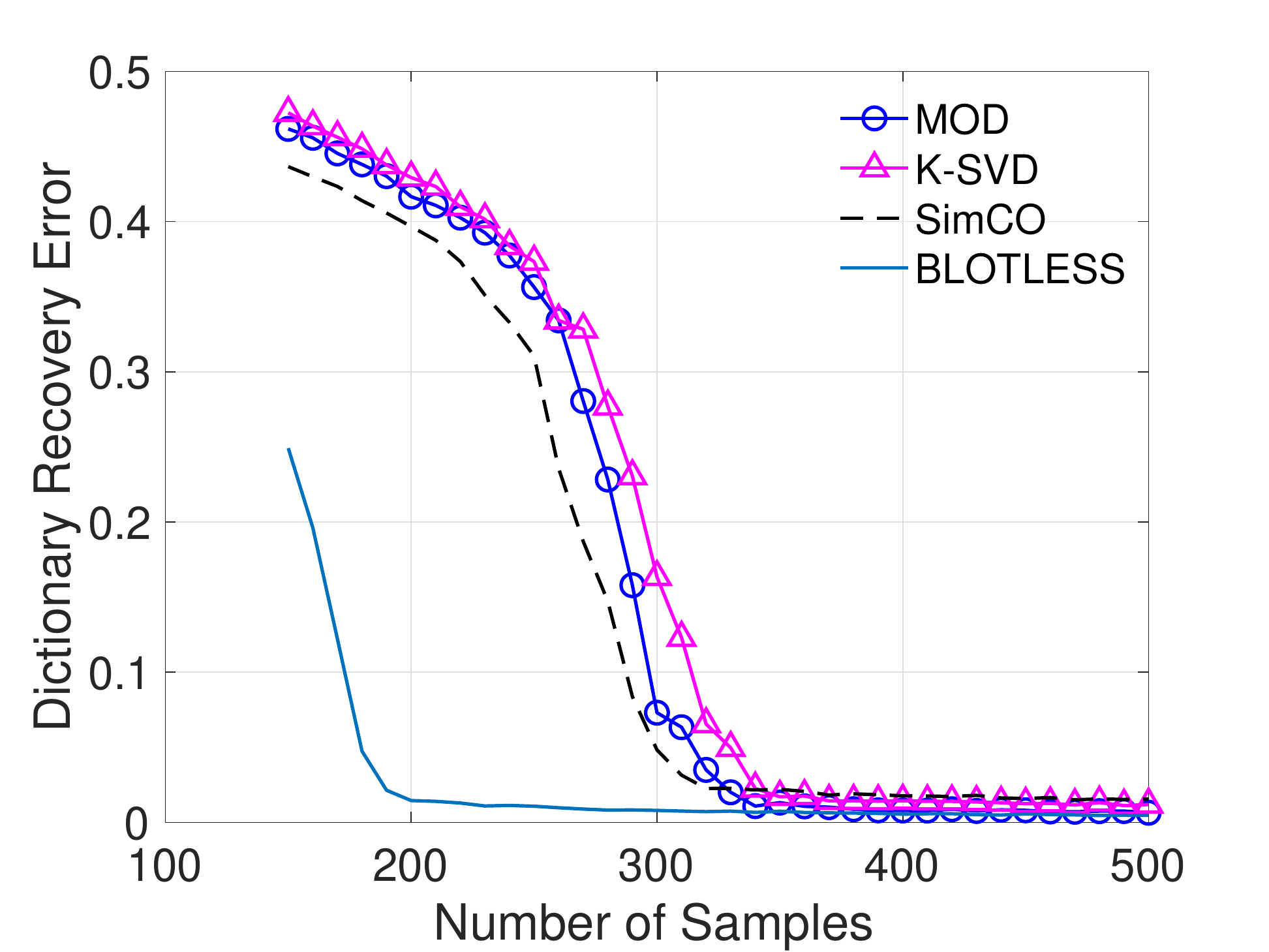}}
	\hspace{0.025cm}
	\subfigure[$m=64$, $l=128$, $\theta=5/128$, $n_{\rm it}=150$.]{
		\label{fig:11} 
		\includegraphics[width=0.23\textwidth]{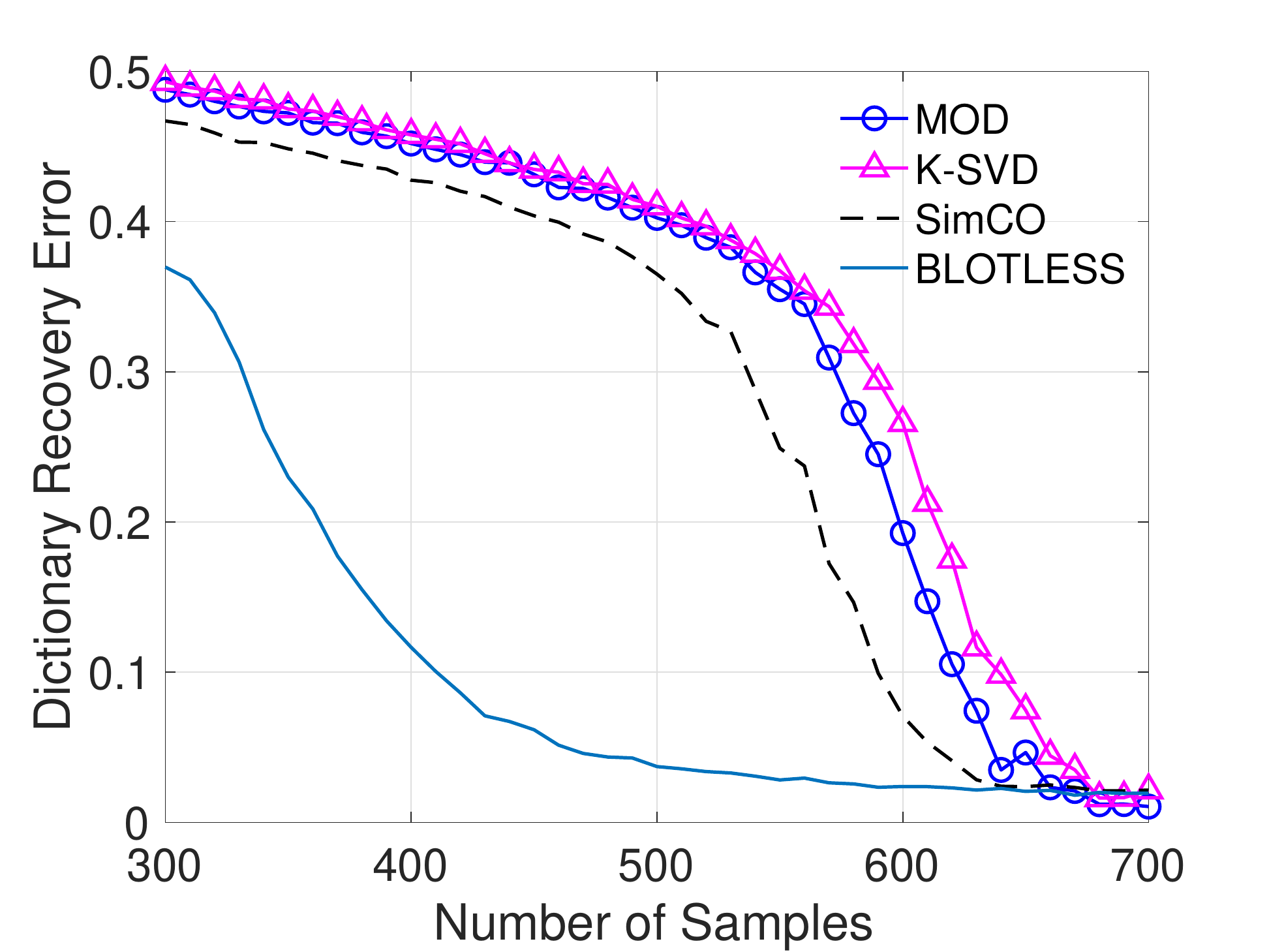}}
	\caption{Comparison of dictionary update methods for the noisy case: SNR is 15dB. Results are averages of 100 trials. $n_{\rm it}$ denoes the number of iterations.}
	\label{fig:DL_noisy}
\end{figure}

Fig. \ref{fig:DL_noisefree} and \ref{fig:DL_noisy} compare the performance of dictionary learning using different dictionary update algorithms. Fig. \ref{fig:DL_noisefree} focus on the noise free case where $\bm{Y}=\bm{D}_0 \bm{X}_0$ and Fig. \ref{fig:DL_noisy} concerns with the noisy case where $\bm{Y}=\bm{D}_0 \bm{X}_0 + \bm{W}$ where $\bm{W}$ is the additive Gaussian noise matrix with i.i.d. entries and the signal-to-noise ratio (SNR) is set to 15dB. Both figures include the cases of complete and over-complete dictionaries. The results presented in Fig. \ref{fig:DL_noisefree} and \ref{fig:DL_noisy} clearly indicate that dictionary learning based on BLOTLESS converges much faster and needs at least $1/3$ less training samples than other benchmark dictionary update methods. 

\begin{figure}[h]
	\centering
	\includegraphics[width=0.36\textwidth]{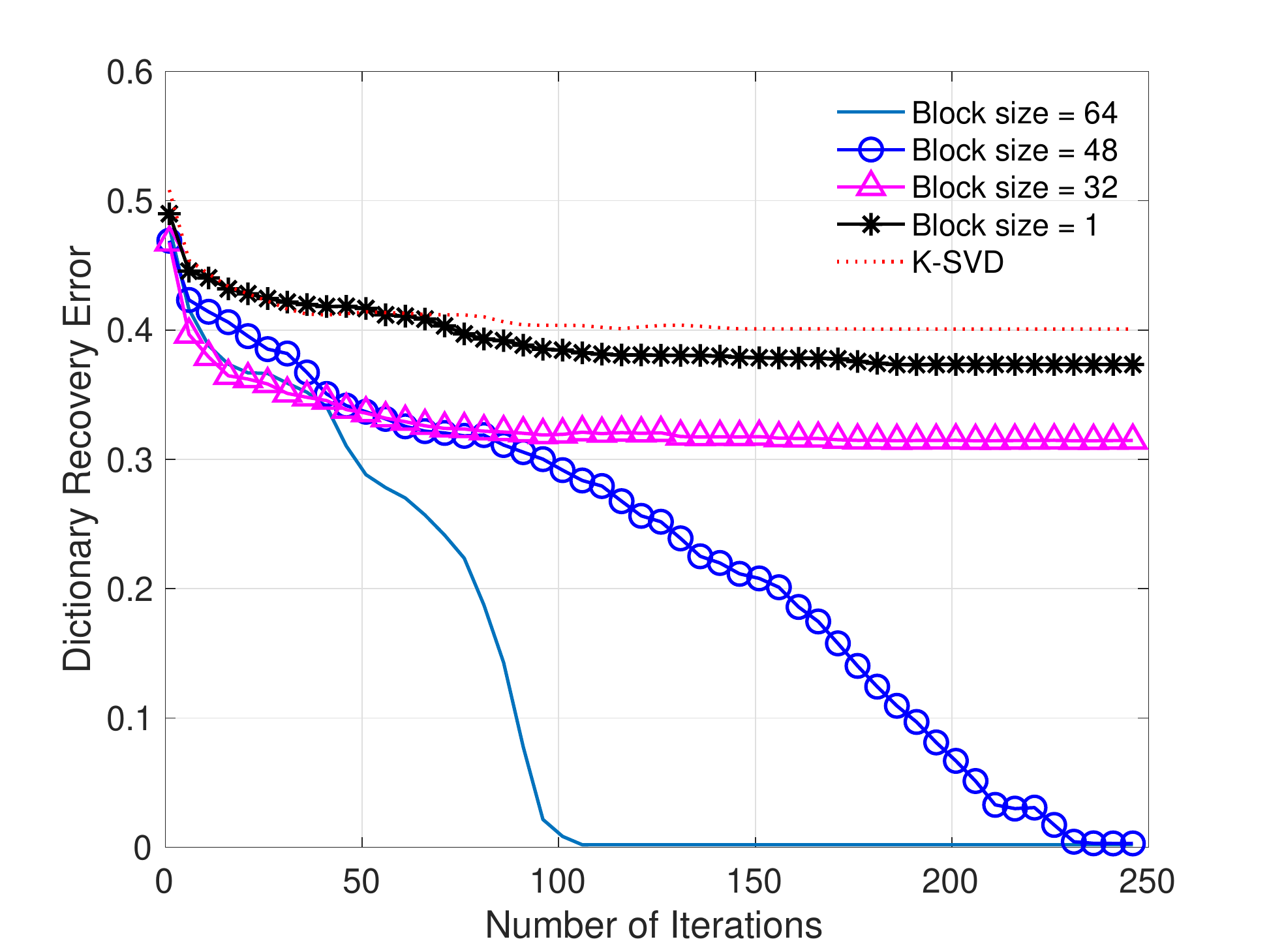}
	\caption{Dictionary learning via BLOTLESS with  different block sizes: $m =l = 64$, $n=200$, $\theta=5/64$.}
	\label{diff_block_size}
\end{figure}

In BLOTLESS update, blocks of the dictionary (sub-dictionaries) are updated sequentially. Fig. \ref{diff_block_size} compares the performance with different block sizes. Note that when each block contains only one dictionary item, the dictionary update problem is the same as that in K-SVD. Hence the performance of K-SVD is added in Fig. \ref{diff_block_size}. Simulations suggest that the larger the dictionary blocks are, the faster the convergence is and the better performance is. The performance of BLOTLESS with block size one is slightly better than that of K-SVD. This is because the TLS step in IterTLS does not  enforce the sparsity pattern and hence better accommodates errors in the estimated sparsity pattern. 

\subsubsection{Real Data \label{subsub:RealData}}

We use the Olivetti Research Laboratory (ORL) face database \cite{samaria1994parameterisation} for dictionary learning and then use the learned dictionary for image denoising. 

For dictionary learning, according to the simulation results in Section \ref{subsub:SyntheticData}, $n=500$ samples of size $8\times 8$ patches from face images \textcolor{black}{are enough} for training a dictionary via BLOTLESS. The parameters used in dictionary learning are $m=64$, $l=128$, $\theta=10/128$, and $n_{\rm it}=20$. After learning a dictionary, image denoising \cite{elad2006image} is performed using test images from the same dataset. The denoising results are shown in Table \ref{tab:denoise_comp}, where four test images are used. In all these four tests, the BLOTLESS method outperforms all other algorithms, which is consistent with these simulations in \ref{subsub:SyntheticData}.

\begin{table*}[h]
	\centering
	\fontsize{6.5}{8}\selectfont
	\caption{Denoising comparison using different learnt dictionaries, where the denoised PSNR (dB) are computed and shown in table.}
	\label{tab:denoise_comp}
	\begin{tabular}{|c|c|c|c|c|c|c|c|c|c|c|c|c|}
		\hline
		Original Image&
		\multicolumn{3}{c|}{\tabincell{c}{\includegraphics[scale=0.55]{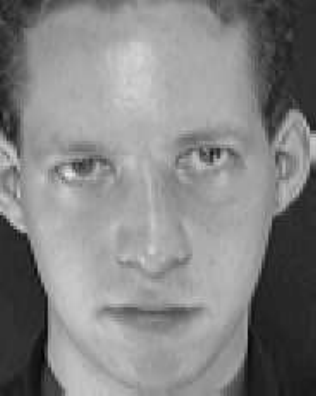}} }&\multicolumn{3}{c|}{\tabincell{c}{\includegraphics[scale=0.55]{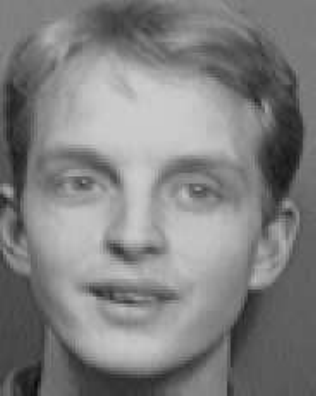}}}&\multicolumn{3}{c|}{\tabincell{c}{\includegraphics[scale=0.55]{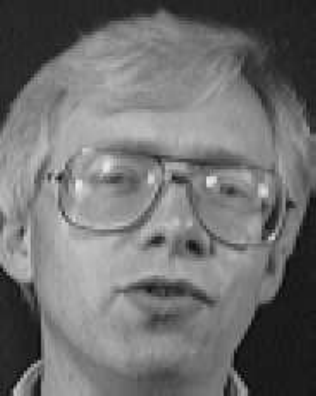}}}&\multicolumn{3}{c|}{\tabincell{c}{ \includegraphics[scale=0.55]{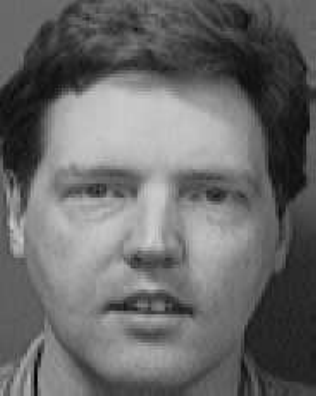}}}\cr\cline{1-13}
		Noisy Image&{\tabincell{c}{28.13 dB\\ \includegraphics[scale=0.3]{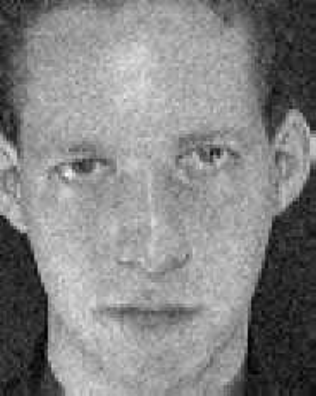}}}&{\tabincell{c}{22.11 dB\\ \includegraphics[scale=0.3]{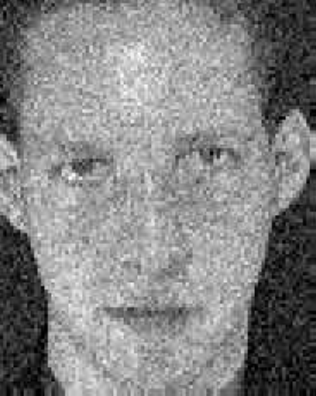}}}&{\tabincell{c}{18.59 dB\\ \includegraphics[scale=0.3]{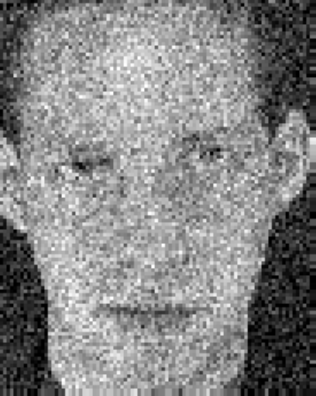}}}&{\tabincell{c}{28.13 dB\\ \includegraphics[scale=0.3]{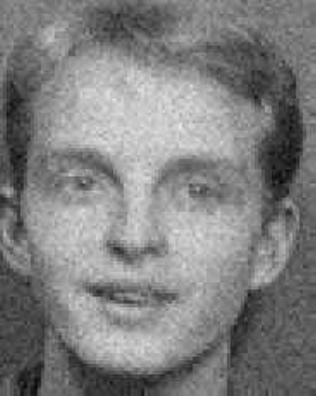}}}&{\tabincell{c}{22.11 dB\\ \includegraphics[scale=0.3]{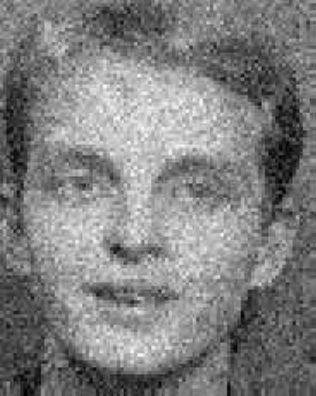}}}&{\tabincell{c}{18.59 dB\\ \includegraphics[scale=0.3]{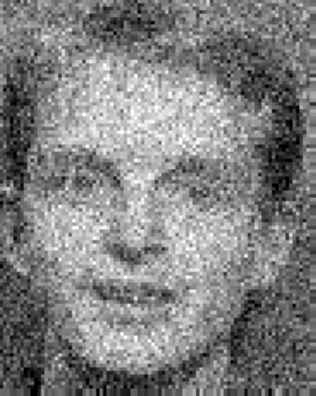}}}&{\tabincell{c}{28.13 dB\\ \includegraphics[scale=0.3]{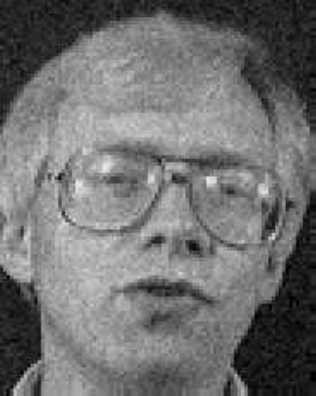}}}&{\tabincell{c}{22.11 dB\\ \includegraphics[scale=0.3]{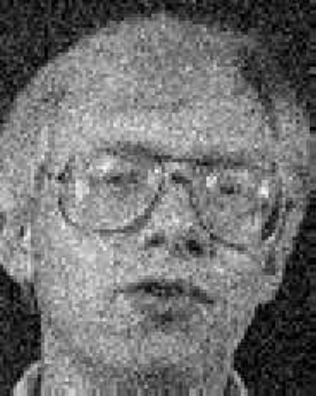}}}&{\tabincell{c}{18.59 dB\\ \includegraphics[scale=0.3]{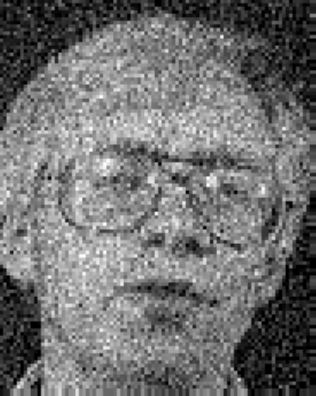}}}&{\tabincell{c}{28.13 dB\\ \includegraphics[scale=0.3]{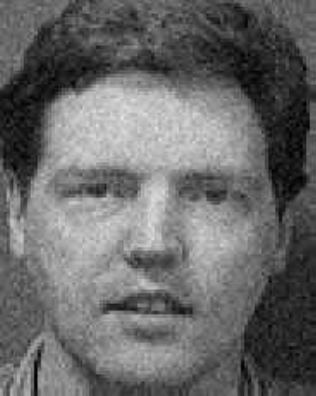}}}&{\tabincell{c}{22.11 dB\\ \includegraphics[scale=0.3]{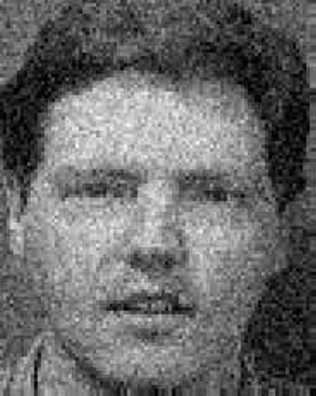}}}&{\tabincell{c}{18.59 dB\\ \includegraphics[scale=0.3]{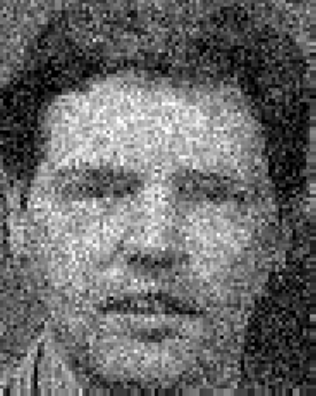}}}\cr
		\hline
		\hline
		MOD&{\tabincell{c}{33.00 dB\\ \includegraphics[scale=0.3]{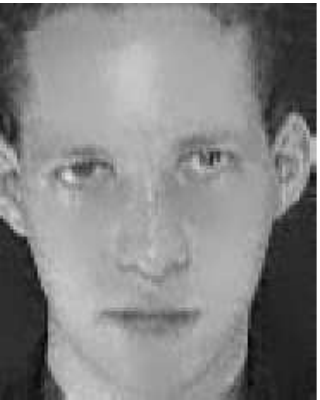}}}&{\tabincell{c}{29.24 dB\\ \includegraphics[scale=0.3]{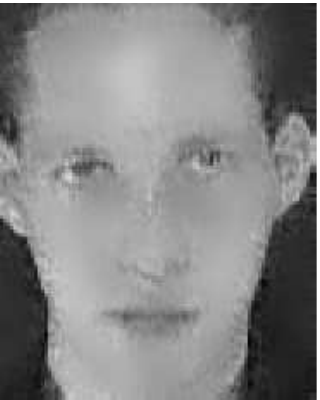}}}&{\tabincell{c}{27.22 dB\\ \includegraphics[scale=0.3]{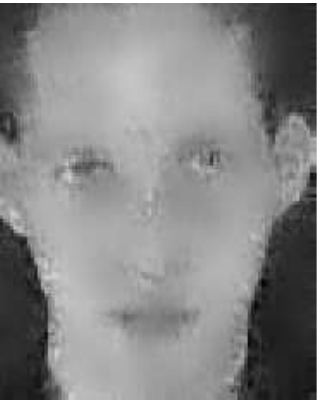}}}&{\tabincell{c}{32.68 dB\\ \includegraphics[scale=0.3]{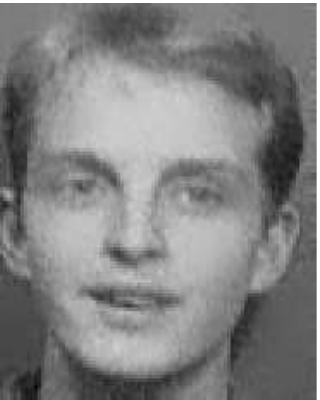}}}&{\tabincell{c}{28.84 dB\\ \includegraphics[scale=0.3]{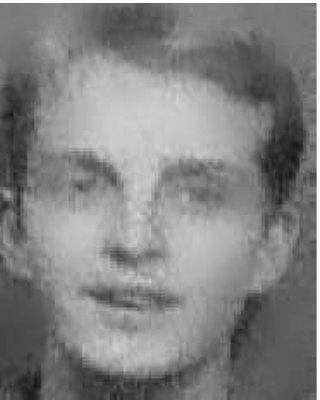}}}&{\tabincell{c}{26.82 dB\\ \includegraphics[scale=0.3]{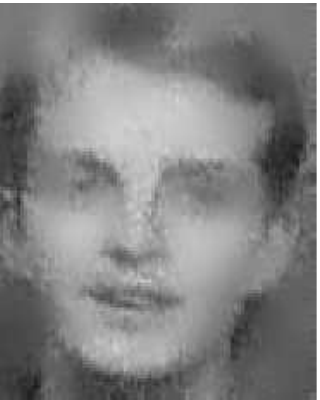}}}&{\tabincell{c}{31.95 dB\\ \includegraphics[scale=0.3]{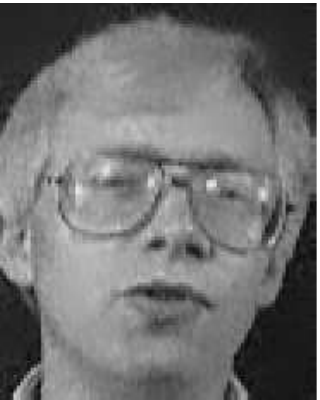}}}&{\tabincell{c}{27.43 dB\\ \includegraphics[scale=0.3]{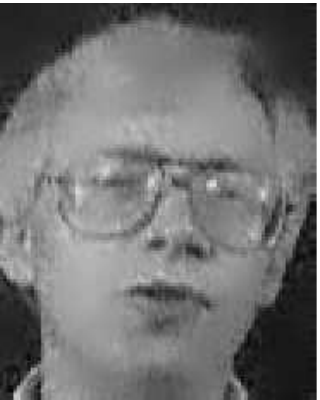}}}&{\tabincell{c}{25.76 dB\\ \includegraphics[scale=0.3]{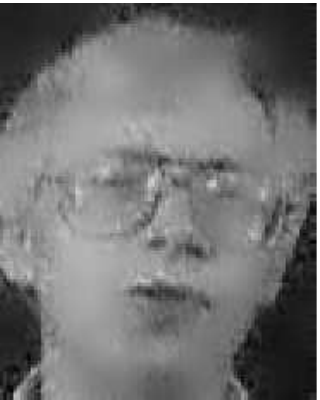}}}&{\tabincell{c}{32.13 dB\\ \includegraphics[scale=0.3]{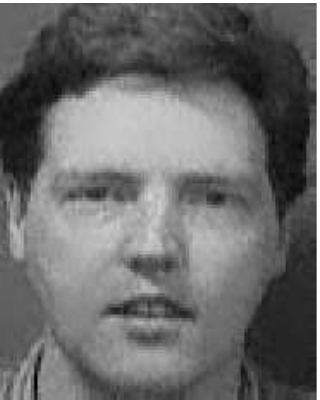}}}&{\tabincell{c}{28.23 dB\\ \includegraphics[scale=0.3]{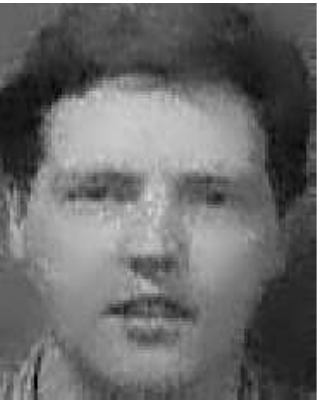}}}&{\tabincell{c}{26.26 dB\\ \includegraphics[scale=0.3]{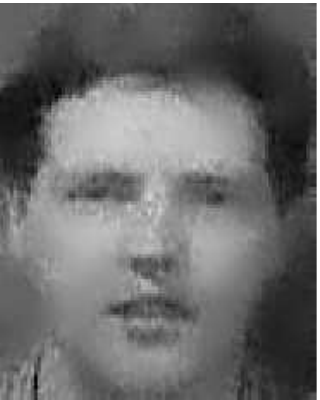}}}\cr\hline
		K-SVD&{\tabincell{c}{32.50 dB\\ \includegraphics[scale=0.3]{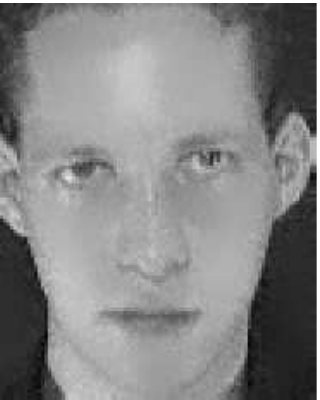}}}&{\tabincell{c}{28.72 dB\\ \includegraphics[scale=0.3]{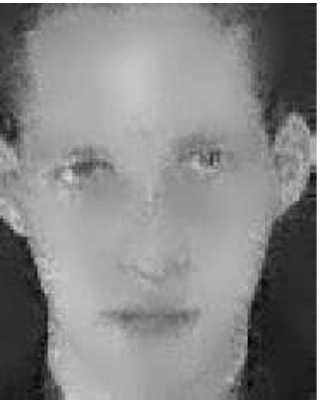}}}&{\tabincell{c}{26.79 dB\\ \includegraphics[scale=0.3]{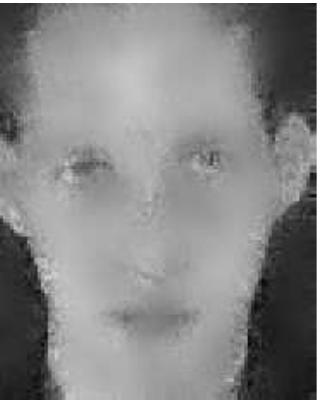}}}&{\tabincell{c}{32.03 dB\\ \includegraphics[scale=0.3]{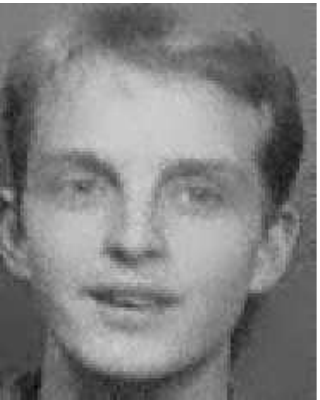}}}&{\tabincell{c}{28.26 dB\\ \includegraphics[scale=0.3]{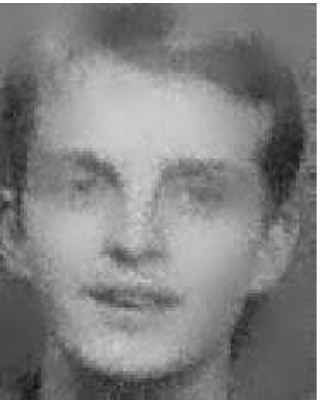}}}&{\tabincell{c}{26.35 dB\\ \includegraphics[scale=0.3]{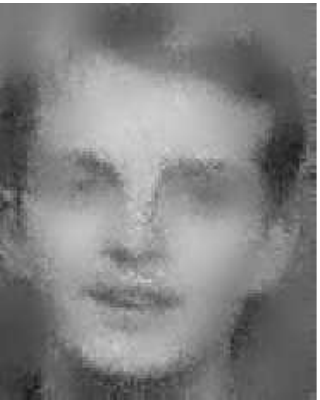}}}&{\tabincell{c}{31.49 dB\\ \includegraphics[scale=0.3]{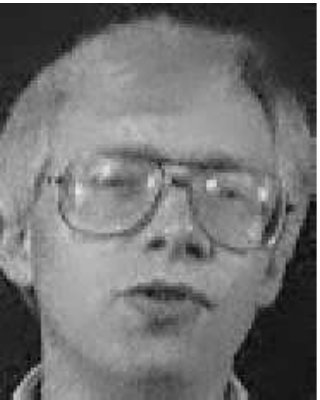}}}&{\tabincell{c}{27.43 dB\\ \includegraphics[scale=0.3]{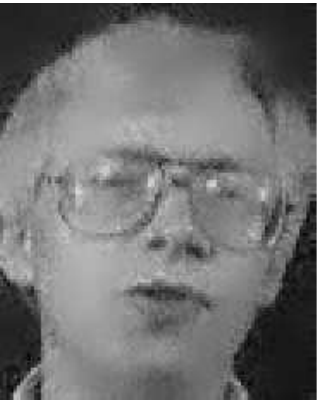}}}&{\tabincell{c}{25.31 dB\\ \includegraphics[scale=0.3]{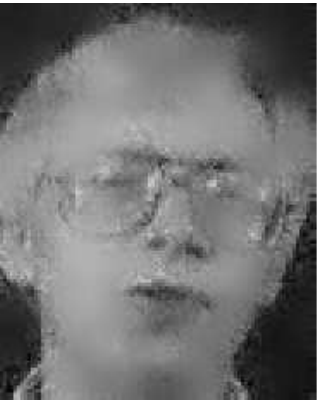}}}&{\tabincell{c}{31.58 dB\\ \includegraphics[scale=0.3]{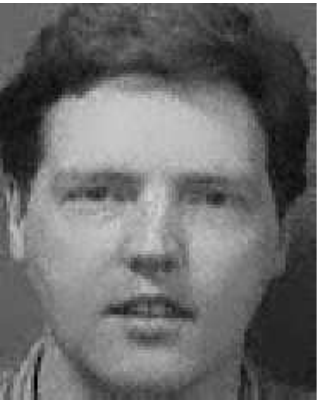}}}&{\tabincell{c}{27.74 dB\\ \includegraphics[scale=0.3]{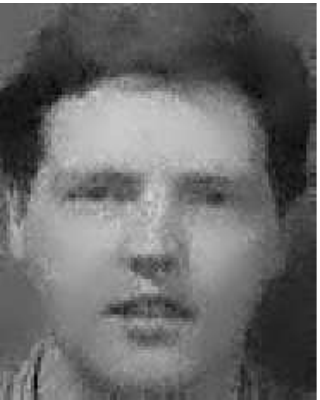}}}&{\tabincell{c}{25.81 dB\\ \includegraphics[scale=0.3]{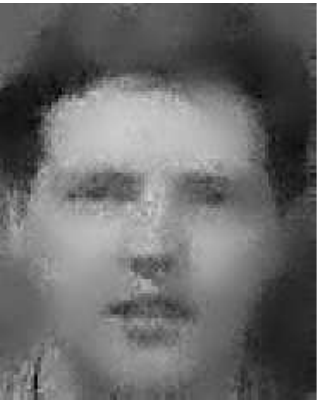}}}\cr\hline
		SimCO&{\tabincell{c}{33.43 dB\\ \includegraphics[scale=0.3]{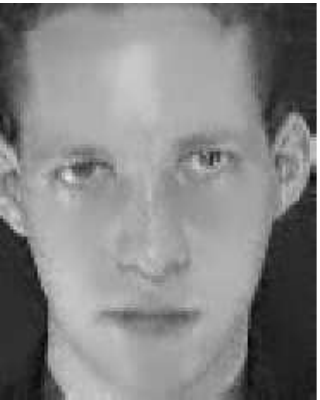}}}&{\tabincell{c}{29.78 dB\\ \includegraphics[scale=0.3]{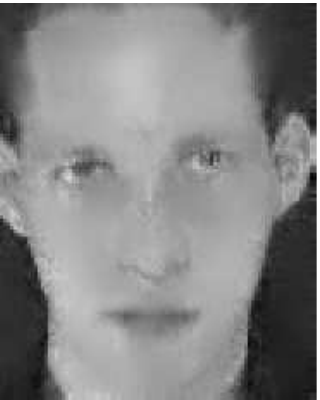}}}&{\tabincell{c}{27.81 dB\\ \includegraphics[scale=0.3]{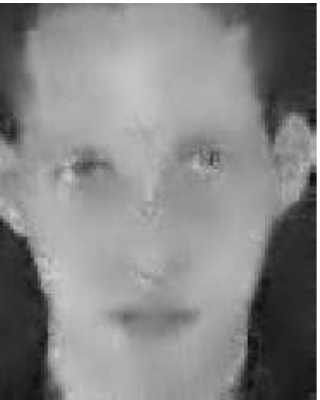}}}&{\tabincell{c}{33.58 dB\\ \includegraphics[scale=0.3]{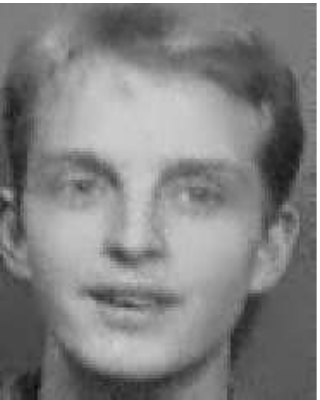}}}&{\tabincell{c}{30.11 dB\\ \includegraphics[scale=0.3]{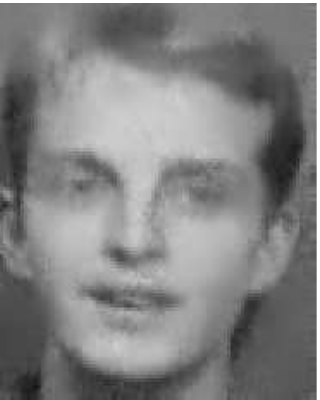}}}&{\tabincell{c}{28.04 dB\\ \includegraphics[scale=0.3]{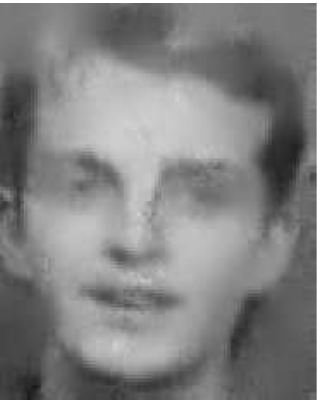}}}&{\tabincell{c}{32.18 dB\\ \includegraphics[scale=0.3]{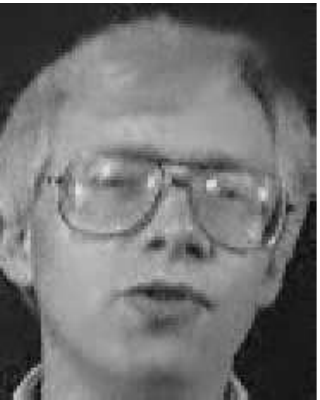}}}&{\tabincell{c}{28.51 dB\\ \includegraphics[scale=0.3]{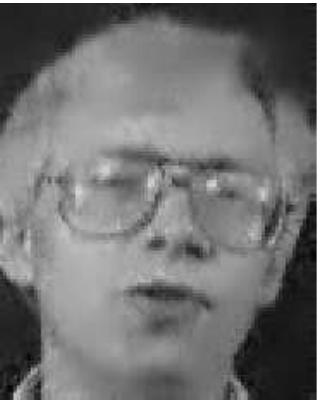}}}&{\tabincell{c}{26.47 dB\\ \includegraphics[scale=0.3]{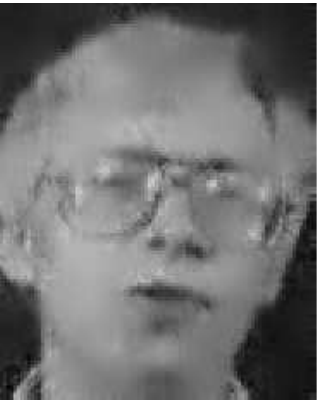}}}&{\tabincell{c}{32.65 dB\\ \includegraphics[scale=0.3]{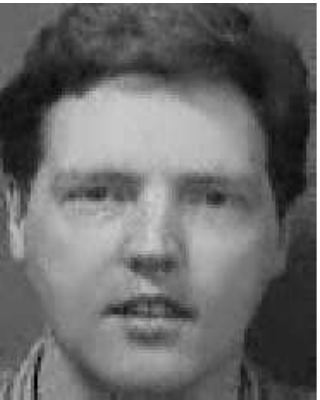}}}&{\tabincell{c}{29.18 dB\\ \includegraphics[scale=0.3]{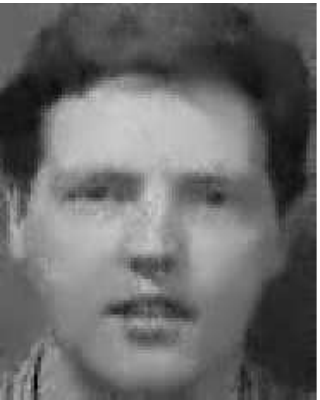}}}&{\tabincell{c}{27.27 dB\\ \includegraphics[scale=0.3]{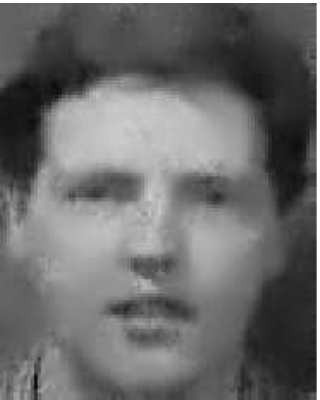}}}\cr\hline
		BLOTLESS&{\tabincell{c}{\textbf{33.67 dB}\\ \includegraphics[scale=0.3]{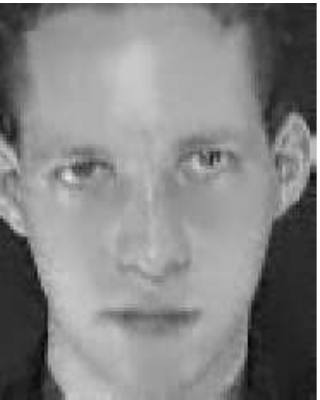}}}&{\tabincell{c}{\textbf{29.90 dB}\\ \includegraphics[scale=0.3]{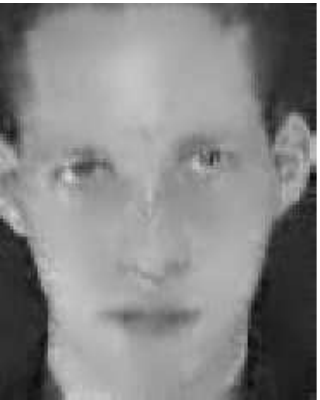}}}&{\tabincell{c}{\textbf{27.95 dB}\\ \includegraphics[scale=0.3]{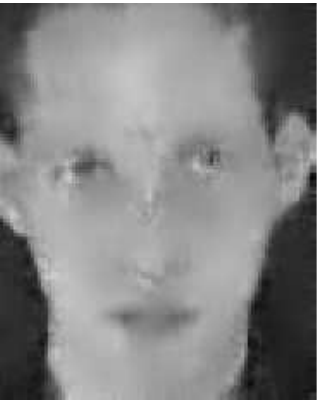}}}&{\tabincell{c}{\textbf{33.91 dB}\\ \includegraphics[scale=0.3]{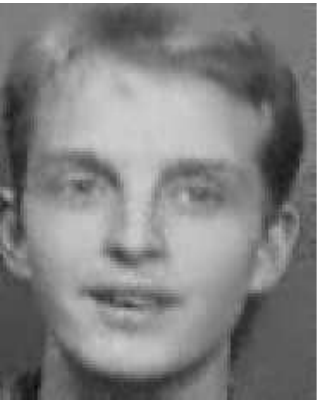}}}&{\tabincell{c}{\textbf{30.33 dB}\\ \includegraphics[scale=0.3]{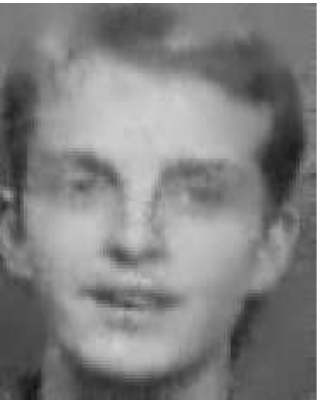}}}&{\tabincell{c}{\textbf{28.20 dB}\\ \includegraphics[scale=0.3]{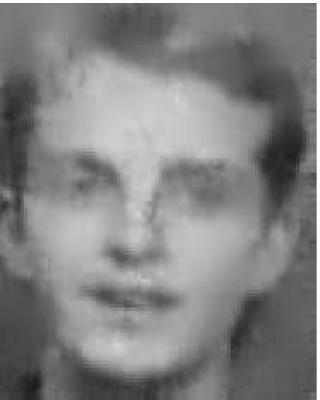}}}&{\tabincell{c}{\textbf{32.38 dB}\\ \includegraphics[scale=0.3]{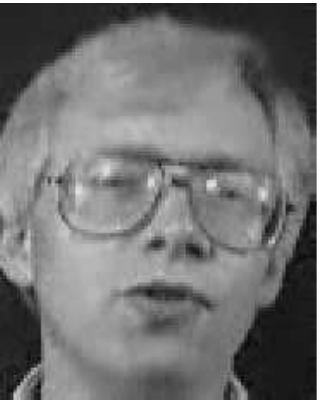}}}&{\tabincell{c}{\textbf{28.67 dB}\\ \includegraphics[scale=0.3]{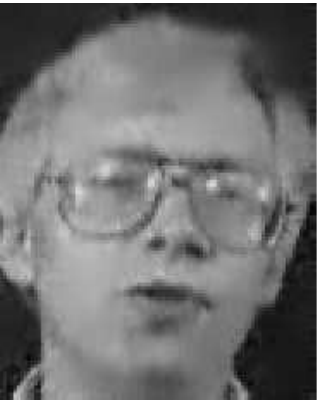}}}&{\tabincell{c}{\textbf{26.66 dB}\\ \includegraphics[scale=0.3]{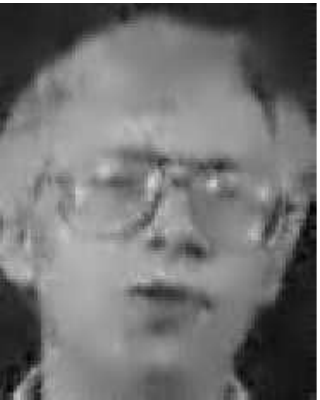}}}&{\tabincell{c}{\textbf{32.88 dB}\\ \includegraphics[scale=0.3]{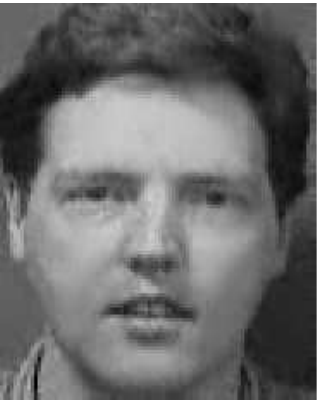}}}&{\tabincell{c}{\textbf{29.38 dB}\\ \includegraphics[scale=0.3]{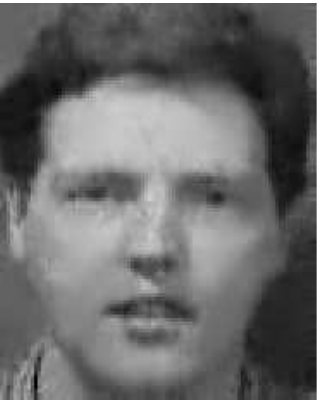}}}&{\tabincell{c}{\textbf{27.45 dB}\\ \includegraphics[scale=0.3]{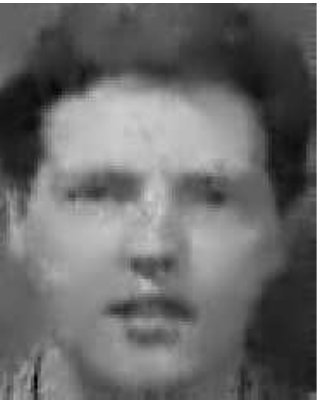}}}\cr\hline
	\end{tabular}
\end{table*}

%

\section{Conclusion \label{sec:Conclusion}}
This paper proposed a BLOTLESS algorithm for dictionary update. It divides the dictionary into sub-dictionaries, each of which is non-overcomplete. Then BLOTLESS updates a sub-dictionary and the corresponding sparse coefficients using least sqaures or total least squares approaches.  Necessary conditions for unique recovery are identified and they hold with high probability when the number of training samples is larger than the derived bounds in Proposition \ref{prop:NumOfSamples}. Simulations show that these bounds match the simulations well, and that BLOTLESS outperforms other benchmark algorithms. \textcolor{black}{One future direction} is to find sufficient bounds for unique recovery and their comparisons to the necessary bounds.

\appendix

\subsection{Proof of Proposition \ref{prop:NumOfSamples}}

The proof needs Hoeffding's inequality\cite{Fisher1963Probability} for Bernoulli random variables, stated below. 

\begin{lem}[Hoeffding's Inequality]
	For $N$ many identical Bernoulli random variables $\{X_{i},i=1,2,...N\}$. Each $X_{i}$ takes the value 1 with probability $p$ and 0 with probability $(1-p)$, then the following Hoeffding's inequality holds
	\begin{equation}
	\Pr\left(\sum_{i=1}^{N}X_{i} \geq (p + \lambda)N\right) \leqslant \exp(-2\lambda^2 N),
	\end{equation}
	where $ \lambda >0$ is a constant number.
\end{lem}

To derive $n_1$, we consider the case that the necessary condition 1 in Proposition \ref{pro:NecessaryBd} fails. That is,
\begin{align*}
\left|\Omega\right| & \ge nm-m^{2}+m \\
& =nm ( \theta+(1-\theta-\frac{m}{n}+\frac{1}{n} ) ).
\end{align*}
Based on Hoeffding's inequality, the probability of this event is upper bounded by 
\begin{equation*}
\exp\left( - 2 \left(1-\theta-\frac{m}{n}+\frac{1}{n}\right)^{2} mn \right).
\end{equation*}
If this probability is smaller than $\epsilon$, it follows that 
\begin{equation*}
\left( \left(1-\theta\right)n - \left(m-1\right) \right)^{2} + \frac{\ln\epsilon}{2m}n \ge 0. 
\end{equation*}
The left hand side of the above inequality is quadratic in $n$, which after some elementary calculations leads to
\begin{align*}
n&\ge n_1 = \frac{m-1}{1-\theta}\left[1-\frac{\ln\epsilon}{4m\left(m-1\right)\left(1-\theta\right)}\right.\\&\quad\left.+\sqrt{\left(1-\frac{\ln\epsilon}{4m\left(m-1\right)\left(1-\theta\right)}\right)^{2}-1}\right].
\end{align*}

The derivation of $n_2$ is similar. Consider the probability that the necessary condition 2 in Proposition \ref{pro:NecessaryBd} fails:
\begin{align*}
& 1-\Pr\left(\forall i \in \left[m\right], \; \left|\Omega_{i}^{c}\right|\ge m-1\right)\\
& =\Pr\left(\exists i\in\left[m\right], \; \left|\Omega_{i}\right| \ge n-m+1\right)\\
& \le m\Pr\left(\left|\Omega_{1}\right|\ge n-m+1\right)\\
& =m\Pr\left(\left|\Omega_{1}\right|\ge n\left(\theta+\left(1-\theta-\frac{m-1}{n}\right)\right)\right),
\end{align*}
where the inequality in the third line follows from the union bound. After applying Hoeffding's inequality and setting the upper bound less than $\epsilon$ we obtain
\begin{equation*}
\left( \left(1-\theta\right)n - \left(m-1\right) \right)^{2} + \frac{\ln\epsilon-\ln m}{2}n \ge 0.
\end{equation*}
It follows that
\begin{align*}
n & \ge n_2 = \frac{m-1}{1-\theta}\left[1-\frac{\ln\epsilon-\ln m}{4\left(m-1\right)\left(1-\theta\right)}\right.\\
& \quad\left.+\sqrt{\left(1-\frac{\ln\epsilon-\ln m}{4\left(m-1\right)\left(1-\theta\right)}\right)^{2}-1}\right].
\end{align*}

To derive $n_3$, we define the following event
\begin{itemize}
	\item $S_{i,i^{\prime}}$: For given $i \neq i^{\prime}$, $\exists j \in[n]$ such that $X_{i,j} = 0$ and $X_{i^{\prime},j} \ne 0$.
\end{itemize}
Then the probability that the necessary condition 3 fails is given
by 
\begin{align*}
& 1-\Pr\left(\cap_{i\ne i^{\prime}} S_{i,i^{\prime}} \right) = \Pr \left( \cup_{i\ne i^{\prime}} S_{i,i^{\prime}}^{c}\right) \\
& \le m\left(m-1\right)\Pr\left(S_{1,2}^{c}\right)\\
& =m\left(m-1\right)\left(1-\theta\left(1-\theta\right)\right)^{n},
\end{align*}
where the inequality in the second line follows from the union bound. If we set this probability to be smaller than $\epsilon$, we obtain 
\begin{align*}
n & \ge n_{3}=\frac{\ln\epsilon-\ln m-\ln\left(m-1\right)}{\ln\left(1-\theta\left(1-\theta\right)\right)}.
\end{align*}

\bibliographystyle{IEEEtran}
\bibliography{TS_LS}

\end{document}